\newif\iflong
\newif\ifshort

\longtrue

\iflong
\else
\shorttrue
\fi

\newcommand\blfootnote[1]{%
  \begingroup
  \renewcommand\thefootnote{}\footnote{#1}%
  \addtocounter{footnote}{-1}%
  \endgroup
}

\documentclass{llncs}

\usepackage{microtype}
\usepackage{graphicx}
\usepackage{subfigure}
\usepackage{tikz}
\usetikzlibrary{shapes, positioning, patterns,decorations.pathreplacing}\tikzstyle{vertex}=[circle, draw, inner sep=0pt, minimum size=4pt,outer sep = 1pt]
\newcommand{\vertex}{\node[vertex,fill]}
\usepackage{booktabs} 

\usepackage{boxedminipage}
\usepackage{cite}
\usepackage{times}
\usepackage{color}
\usepackage[linesnumbered,boxed]{algorithm2e}
\usepackage{graphicx}
\usepackage{latexsym}
\usepackage[hidelinks]{hyperref}
\usepackage{amsmath,verbatim}
\usepackage{amssymb}
\usepackage{enumerate}
\usepackage{todonotes}
\usepackage[draft,author=]{fixme}

\newcommand{\cut}{\operatorname{cut}}
\newcommand{\cc}[1]{{\mbox{\textnormal{\textsf{#1}}}}\xspace} 
\newcommand{\NP}{\cc{NP}}
\newcommand{\FPT}{\cc{FPT}}
\newcommand{\XP}{\cc{XP}}
\newcommand{\W}[1]{\ensuremath {\cc{W}[#1]}}

\newcommand{\Nat}{\mathbb{N}}
\newcommand{\bigoh}{\mathcal{O}}
\newcommand{\XXX}{\mathcal{X}}

\newtheorem{observation}{Observation}

\title{Edge-Cut Width: An Algorithmically Driven Analogue of Treewidth Based on Edge Cuts}

\author{Cornelius Brand \and Esra Ceylan \and Robert Ganian \and Christian Hatschka \and Viktoriia Korchemna}
\institute{Algorithms and Complexity Group, TU Wien, Vienna, Austria}

\usepackage{amsmath,amssymb}
\usepackage{todonotes}
\usepackage{enumitem}

\newcommand{\width}{edge-cut width}

\newcommand{\ecw}{\widthshort}
\newcommand{\vcn}{\operatorname{vcn}}
\newcommand{\fen}{\operatorname{fen}}
\newcommand{\widthshort}{\operatorname{ecw}}
\newcommand {\extwidthshort}{\widthshort^\ast}

\newcommand{\EDP}{\textsc{EDP}}
\newcommand{\PP}{P}
\newcommand{\QQ}{Q}
\newcommand{\adh}{\operatorname{adh}}
\newcommand{\tcw}{\operatorname{tcw}}
\newcommand{\tor}{\operatorname{tor}}
\newcommand{\tw}{\operatorname{tw}}
\newcommand{\degtw}{\operatorname{degtw}}
\newcommand{\loc}{\operatorname{loc}}

\newcommand{\fut}{\texttt{future}}
\newcommand{\past}{\texttt{past}}

\newcommand{\cost}{\texttt{cost}}
\newcommand{\col}{\texttt{col}}

\newcommand{\outgoing}{\texttt{Outgoing}}
\newcommand{\donate}{\texttt{Donate}}
\newcommand{\suc}{\emph{succ}}
\usepackage{needspace}

\newcommand{\pbDef}[3]{%
\noindent
\begin{center}
\begin{boxedminipage}{0.98 \columnwidth}
#1\\[5pt]
\begin{tabular}{l p{0.85 \columnwidth}}
Input: & #2\\
Question: & #3
\end{tabular}
\end{boxedminipage}
\end{center}
}

\begin{document}
\maketitle
\begin{abstract}
Decompositional parameters such as treewidth are commonly used to obtain fixed-parameter algorithms for \NP-hard graph problems. For problems that are \W{1}-hard parameterized by treewidth, a natural alternative would be to use a suitable analogue of treewidth that is based on edge cuts instead of vertex separators. While tree-cut width has been coined as such an analogue of treewidth for edge cuts, its algorithmic applications have often led to disappointing results: out of twelve problems where one would hope for fixed-parameter tractability parameterized by an edge-cut based analogue to treewidth, eight were shown to be \W{1}-hard parameterized by tree-cut width. 

As our main contribution, we develop an edge-cut based analogue to treewidth called \width. Edge-cut width is, intuitively, based on measuring the density of cycles passing through a spanning tree of the graph. Its benefits include not only a comparatively simple definition, but mainly that it has interesting algorithmic properties: it can be computed by a fixed-parameter algorithm, and it yields fixed-parameter algorithms for all the aforementioned problems where tree-cut width failed to do so.
\blfootnote{Cornelius Brand, Robert Ganian and Viktoriia Korchemna gratefully acknowledge support from the Austria Science Foundation (FWF, Project Y1329).}
\end{abstract}

\section{Introduction}
While the majority of computational problems on graphs are intractable, in most cases it is possible to exploit the structure of the input graphs to circumvent this intractability. This basic fact has led to the extensive study of a broad hierarchy of decompositional graph parameters (see, e.g., Figure~1 in~\cite{BodlaenderJK13}), where for individual problems of interest the aim is to pinpoint which parameters can be used to develop fixed-parameter algorithms for the problem. Treewidth~\cite{RobertsonSeymour86} is by far the most prominent parameter in the hierarchy, and it is known that many problems of interest are fixed-parameter tractable when parameterized by treewidth; some of these problem can even be solved efficiently on more general parameters such as rank-width~\cite{Oum08,GanianHlineny10} or other decompositional parameters above treewidth in the hierarchy~\cite{Bonnet0TW20}. However, in this article we will primarily be interested in problems that lie on the other side of this spectrum: those which remain intractable when parameterized by treewidth.

Aside from non-decompositional parameters\footnote{We view a parameter as decompositional if it is tied to a well-defined graph decomposition; all decompositional parameters are closed under the disjoint union operation of graphs.} such as the vertex cover number~\cite{FellowsLokshtanovMisraRS08,Ganian15} or feedback edge number~\cite{BentertHHKN20,GanianO21,GolovachKKL22}, 
the most commonly applied parameters for problems which are not fixed-parameter tractable with respect to treewidth are tied to the existence of small vertex separators. One example of such a parameter is treedepth~\cite{NesetrilOssonademendez12}, which has by now found numerous applications in diverse areas of computer science~\cite{GutinJW16,GanianO18,NederlofPSW20}. 
%
 %
An alternative approach is to use a decompositional parameter that is inherently tied to edge-cuts---in particular, tree-cut width~\cite{MarxWollan14,Wollan15}. 

Tree-cut width was discovered by Wollan, who described it as a variation of tree decompositions based on edge cuts instead of vertex separators~\cite{Wollan15}.
But while it is true that ``tree-cut decompositions share many of the natural properties of tree decompositions''~\cite{MarxWollan14}, from the perspective of algorithmic design tree-cut width seems to behave differently than an edge-cut based alternative to treewidth. To illustrate this, we note that tree-cut width is a parameter that lies between treewidth and treewidth plus maximum degree (which may be seen as a ``heavy-handed'' parameterization that enforces small edge cuts) in the parameter hierarchy~\cite{Ganian0S15,KimOPST18}. There are numerous problems which are \W{1}-hard (and sometimes even \NP-hard) w.r.t.\ treewidth but fixed-parameter tractable w.r.t.\ the latter parameterization, and the aim would be to have an edge-cut based parameter that can lift this fixed-parameter tractability towards graphs of unbounded degree. 

Unfortunately, out of twelve problems with these properties where a tree-cut width parameterization has been pursued so far, only four are fixed-parameter tractable~\cite{Ganian0S15,GanianKO21} while eight turn out to be \W{1}-hard~\cite{Ganian0S15,GozupekOPSS17,BredereckHKN19,GanianO21,GanianKorchemna21}. 
The most appalling example of the latter case is the well-established \textsc{Edge Disjoint Paths} (\EDP) problem: \textsc{Vertex Disjoint Paths} is a classical example of a problem that is \FPT\ parameterized by treewidth, 
and one should by all means expect a similar outcome for \EDP\ parameterized by the analogue of treewidth based on edge cuts~\cite{GanianOR21,GanianO21}. But if \EDP\ is \W{1}-hard parameterized by tree-cut width, what is the algorithmic analogue of treewidth for edge cuts?
Here, we attempt to answer to this question through the notion of \width. 


\smallskip
\noindent \textbf{Contribution.} \quad
Edge-cut width is an edge-cut based decompositional parameter which has a surprisingly streamlined definition: instead of specialized decompositions such as those employed by treewidth, clique-width or tree-cut width, the ``decompositions'' for edge-cut width are merely spanning trees (or, in case of disconnected graphs, maximum spanning forests). To define \width\ of a spanning tree $T$, we observe that for each edge in $G-T$ there is a unique path in $T$ connecting its endpoints, and the \width\ of $T$ is merely the maximum number of such paths that pass through any particular vertex in $T$; as usual, the \width\ of $G$ is then the minimum width of a spanning tree (i.e., decomposition).

After introducing \width, establishing some basic properties of the parameter and providing an in-depth comparison to tree-cut width, we show that the parameter has surprisingly useful algorithmic properties. As our first task, we focus on the problem of computing \width\ along with a suitable decomposition.
%
%
This is crucial, since we will generally need to compute an \width\ decomposition before we can use the parameter to solve problems of interest. As our first algorithmic result, we leverage the connection of \width\ to spanning trees of the graph to obtain an explicit fixed-parameter algorithm for computing \width\ decompositions.
This compares favorably to tree-cut width, for which only an explicit $2$-approximation fixed-parameter algorithm~\cite{KimOPST18} and a non-constructive fixed-parameter algorithm~\cite{GiannopoulouKRT19} are known.

Finally, we turn to the algorithmic applications of \width. Recall that among the twelve problems where a parameterization by tree-cut width had been pursued,
eight were shown to be \W{1}-hard parameterized by tree-cut width: \textsc{List Coloring}~\cite{Ganian0S15}, \textsc{Precoloring Extension}~\cite{Ganian0S15}, \textsc{Boolean Constraint Satisfaction}~\cite{Ganian0S15}, \textsc{Edge Disjoint Paths}~\cite{GanianO21}, \textsc{Bayesian Network Structure Learning}~\cite{GanianKorchemna21}, \textsc{Polytree Learning}~\cite{GanianKorchemna21}, \textsc{Minimum Changeover Cost Arborescence}~\cite{GozupekOPSS17}, and \textsc{Maximum Stable Roommates with Ties and Incomplete Lists}~\cite{BredereckHKN19}. Here, we follow up on previous work by showing that \emph{all} of these problems are fixed-parameter tractable when parameterized by \width. We obtain our algorithms using a new dynamic programming framework for \width, which can also be adapted for other problems of interest.

\smallskip
\noindent \textbf{Related Work.}\quad
The origins of \width\ lie in the very recent work of Ganian and Korchemna on learning polytrees and Bayesian networks~\cite{GanianKorchemna21}, who discovered an equivalent parameter when attempting to lift the fixed-parameter tractability of these problems to a less restrictive parameter than the feedback edge number\footnote{The authors originally used the name ``local feedback edge number''.}. That same work also showed that computing edge-cut width can be expressed in Monadic Second Order Logic which implies fixed-parameter tractability, but obtaining an explicit fixed-parameter algorithm for computing optimal decompositions was left as an open question.

As far as the authors are aware, there are only four problems for which it is known that fixed-parameter tractability can be lifted from the parameterization by ``maximum degree plus treewidth'' to tree-cut width. These are \textsc{Capacitated Vertex Cover}~\cite{Ganian0S15}, \textsc{Capacitated Dominating Set}~\cite{Ganian0S15}, \textsc{Imbalance}~\cite{Ganian0S15} and \textsc{Bounded Degree Vertex Deletion}~\cite{GanianKO21}. Additionally, Gozupek et al.~\cite{GozupekOPSS17} showed that the \textsc{Minimum Changeover Cost Arborescence} problem is fixed-parameter tractable when parameterized by a special, restricted version of tree-cut width where one essentially requires the so-called \emph{torsos} to be stars.

\ifshort
\smallskip
\noindent \emph{Statements where proofs or more details are provided in the appendix are marked with $\star$}.
\fi

\section{Preliminaries}
\ifshort
We use standard terminology for graph theory, see for
instance~\cite{Diestel12}, and assume basic familiarity with the parameterized complexity paradigm including, in particular, the notions of \emph{fixed-parameter tractability} and \W{1}-\emph{hardness}~\cite{DowneyFellows13,CyganFKLMPPS15}. Let $\Nat$ denote the set of natural numbers including zero.
We use $[i]$ to
denote the set $\{0,1,\dots,i\}$. 

Given two graph parameters $\alpha,\beta: G\rightarrow \Nat$, we say that $\alpha$ \emph{dominates} $\beta$ if there exists a function $p$ such that for each graph $G$, $\alpha(G)\leq p(\beta(G))$. For a vertex set $Y$, we use $N(Y)$ to denote the set of all vertices that are outside of $Y$ and have a neighbor in $Y$.

%
%
\fi

\iflong
We use standard terminology for graph theory, see for
instance~\cite{Diestel12}.
Given a graph $G$, we let $V(G)$ denote its vertex set and $E(G)$ its
edge set.
The \emph{(open) neighborhood} of a vertex $x \in V(G)$ is the set $\{y\in V(G):xy\in E(G)\}$ and is denoted by $N_G(x)$. For a vertex subset $X$, the neighborhood of $X$ is defined as $\bigcup_{x\in X} N_G(x) \setminus X$ and denoted by $N_G(X)$; we drop the subscript if the graph is clear from the context. \emph{Contracting} an edge $\{a,b\}$ is the operation of replacing vertices $a,b$ by a new vertex whose neighborhood is $(N(a)\cup N(b))\setminus \{a,b\}$.
For a vertex set $A$ (or edge set $B$), we use $G-A$ ($G-B$) to denote the graph obtained from
$G$ by deleting all vertices in $A$ (edges in $B$), and we use $G[A]$ to denote the
\emph{subgraph induced on} $A$, i.e., $G- (V(G)\setminus A)$. 

A \emph{forest} is a graph without cycles, and an edge set $X$ is a
\emph{feedback edge set} if $G-X$ is a forest. 
We use $[i]$ to
denote the set $\{0,1,\dots,i\}$. 

Given two graph parameters $\alpha,\beta: G\rightarrow \Nat$, we say that $\alpha$ \emph{dominates} $\beta$ if there exists a function $p$ such that for each graph $G$, $\beta(G)\leq p(\alpha(G))$. 


\subsection{Parameterized Complexity}

	A \emph{parameterized problem} $\PP$ is a subset of $\Sigma^* \times \Nat$ for some finite alphabet $\Sigma$. Let $L\subseteq \Sigma^*$ be a classical decision problem for a finite alphabet, and let $p$ be a non-negative integer-valued function defined on $\Sigma^*$. Then $L$ \emph{parameterized by} $p$ denotes the parameterized problem $\{(x,p(x)) | x\in L \}$ where $x\in \Sigma^*$.  For a problem instance $(x,k) \in \Sigma^* \times \Nat$ we call $x$ the main part and $k$ the parameter.  
	A parameterized problem~$\PP$ is \emph{fixed-parameter   tractable} (FPT in short) if a given instance $(x, k)$ can be solved in time  $f(k) \cdot |x|^{\bigoh(1)}$ where $f$ is an arbitrary computable function of $k$. We call algorithms running in this time \emph{fixed-parameter algorithms}. 
	
	Parameterized complexity classes are defined with respect to {\em fpt-reducibility}. A parameterized problem $P$ is {\em fpt-reducible} to $Q$ if in time $f(k)\cdot |x|^{\bigoh(1)}$, one can transform an instance $(x,k)$ of $\PP$ into an instance $(x',k')$ of $\QQ$ such that $(x,k)\in \PP$ if and only if $(x',k')\in \QQ$, and $k'\leq g(k)$, where $f$ and $g$ are computable functions depending only on $k$. 
	Owing to the definition, if $\PP$ fpt-reduces to $\QQ$
	and $\QQ$ is fixed-parameter tractable then $P$ is fixed-parameter
	tractable as well. 
	Central to parameterized complexity is the following hierarchy of complexity classes, defined by the closure of canonical problems under fpt-reductions:
	\[\FPT \subseteq \W{1} \subseteq \W{2} \subseteq \cdots \subseteq \XP.\] All inclusions are believed to be strict. In particular, $\FPT\neq \W{1}$ under the Exponential Time Hypothesis. 
	
	The class $\W{1}$ is the analog of $\NP$ in parameterized complexity. 
	A major goal in parameterized complexity is to distinguish between parameterized problems which are in $\FPT$ and those which are $\W{1}$-hard, i.e., those to which every problem in $\W{1}$ is fpt-reducible. There are many problems shown to be complete for $\W{1}$, or equivalently $\W{1}$-complete, including the {\sc Multi-Colored Clique (MCC)} problem~\cite{DowneyFellows13}. We refer the reader to the respective monographs~\cite{DowneyFellows13,CyganFKLMPPS15} for an in-depth
	introduction to parameterized complexity.
	\fi

\iflong
\subsection{Treewidth}
\fi
\ifshort
\smallskip 
\noindent \textbf{Treewidth.}\quad
\fi
\emph{Treewidth}~\cite{RobertsonSeymour86} is a fundamental graph parameter that has found a multitude of algorithmic applications throughout computer science. 
\begin{definition}
A \emph{tree decomposition} of a graph $G$ is a pair $(T,\{\beta_t\}_{t \in V(T)})$, where $T$ is a tree, and each node $t \in V(T)$ is associated with a \emph{bag} $\beta_t \subseteq V(G)$, satisfying the following conditions:
\begin{enumerate}
	\item Every vertex of $G$ appears in some bag of $T$.
	\item Every edge of $G$ is contained as a subset in some bag of $T$.
	\item For every vertex $v \in V(G)$, the set of nodes $t \in V(T)$ such that $v \in \beta_t$ holds is connected in $T$.
\end{enumerate}
The \emph{width} of a tree decomposition is defined as $\max_t |\beta_t| - 1$,
and the \emph{treewidth} $\tw(G)$ of $G$ is defined as the minimum width of any of its tree decompositions.
\end{definition}
\iflong
For our algorithms, it will be useful to make some additional assumptions on the tree decomposition. 
\begin{definition}
A tree decomposition $(T,\{\beta_t\}_{t\in V(T)})$ is called \emph{nice} if is satisfies the following:
\begin{enumerate}
\item $T$ has a distinguished \emph{root} $r \in V(T)$ with $\beta_r = \emptyset$.
\item Every node of $T$ has at most two children.
\item For every node $t$ of $T$ with two children $s,u$ it holds that $\beta_t = \beta_{s} = \beta_u$. These nodes are called \emph{join-nodes}.
\item For every node $t$ of $T$ with exactly one child $s$, there is a vertex $v \in V(G)$ such that either $\beta_t - \beta_u = \{v\}$, in which case we call $t$ an \emph{introduce-node}, or $\beta_u - \beta_t = \{v\}$, in which case we call $t$ a \emph{forget-node}. 
\iflong We call $v$ the vertex \emph{introduced} (resp. \emph{forgotten}) at $t$.\fi
\item Every node that has no children is called a \emph{leaf-node} of $T$ and $\beta_t = \emptyset$ must hold.
\end{enumerate}
\end{definition}
It is known that every tree decomposition can be converted into a nice one of the same width in linear time.
\fi

\iflong
\subsection{Tree-cut Width}
\label{sub:tcw}
\fi
\ifshort
\smallskip
\noindent \textbf{Tree-cut Width.}\quad
\fi
The notion of tree-cut decompositions was introduced by Wollan~\cite{Wollan15}, see also~\cite{MarxWollan14}.
A family of subsets $X_1, \ldots, X_{k}$ of $X$ is a {\em near-partition} of $X$ if they are pairwise disjoint and $\bigcup_{i=1}^{k} X_i=X$, allowing the possibility of $X_i=\emptyset$.  

\begin{definition}
	A {\em tree-cut decomposition} of $G$ is a pair $(T,\mathcal{X})$ which consists of a rooted tree $T$ and a near-partition $\mathcal{X}=\{X_t\subseteq V(G): t\in V(T)\}$ of $V(G)$. A set in the family~$\mathcal{X}$ is called a {\em bag} of the tree-cut decomposition. 
\end{definition}

%

For any node $t$ of $T$ other than the root $r$, let $e(t)=ut$ be the unique edge incident to $t$ on the path to $r$. Let $T_u$ and $T_t$ be the two connected components in $T-e(t)$ which contain $u$ and $t$, respectively. Note that $(\bigcup_{q\in T_u} X_q, \bigcup_{q\in T_t} X_q)$ is a near-partition of $V(G)$, and we use $E_t$ to denote the set of edges with one endpoint in each part. We define the {\em adhesion} of $t$ ($\adh(t)$) as $|E_t|$;
we explicitly set $\adh(r)=0$ and $E(r)=\emptyset$.

The {\em torso} of a tree-cut decomposition $(T,\mathcal{X})$ at a node $t$, written as $H_t$, is the graph obtained from $G$ as follows. If $T$ consists of a single node $t$, then the torso of $(T,\mathcal{X})$ at $t$ is $G$. Otherwise, let $T_1, \ldots , T_{\ell}$ be the connected components of $T-t$. For each $i=1,\ldots , \ell$, the vertex set $Z_i\subseteq V(G)$ is defined as the set $\bigcup_{b\in V(T_i)}X_b$. The torso $H_t$ at $t$ is obtained from $G$ by {\em consolidating} each vertex set $Z_i$ into a single vertex $z_i$ (this is also called \emph{shrinking} in the literature). Here, the operation of consolidating a vertex set~$Z$ into $z$ is to substitute $Z$ by $z$ in $G$, and for each edge $e$ between $Z$ and $v\in V(G)\setminus Z$, adding an edge $zv$ in the new graph. We note that this may create parallel edges.

The operation of {\em suppressing} (also called \emph{dissolving} in the literature) a vertex $v$ of degree at most $2$ consists of deleting~$v$, and when the degree is two, adding an edge between the neighbors of $v$. Given a connected graph $G$ and  $X\subseteq V(G)$, let the {\em 3-center} of $(G,X)$ be the unique graph obtained from $G$ by exhaustively suppressing vertices in $V(G) \setminus X$ of degree at most two. Finally, for a node $t$ of $T$, we denote by $\tilde{H}_t$ the 3-center of $(H_t,X_t)$, where $H_t$ is the torso of $(T,\mathcal{X})$ at $t$. 
Let the \emph{torso-size} $\tor(t)$ denote $|\tilde{H}_t|$. 

\begin{definition}
	The width of a tree-cut decomposition $(T,\mathcal{X})$ of $G$ is $\max_{t\in V(T)}\{ \adh(t),$ $\tor(t) \}$. The tree-cut width of $G$, or $\tcw(G)$ in short, is the minimum width of $(T,\mathcal{X})$ over all tree-cut decompositions $(T,\mathcal{X})$ of $G$.
\end{definition}

\iflong

Without loss of generality, we shall assume that $X_r=\emptyset$.
We conclude this subsection with some notation related to tree-cut decompositions. 
Given a tree node $t$, let $T_t$ be the subtree of $T$ rooted at $t$. Let $Y_t=\bigcup_{b\in V(T_t)} X_b$, and let $G_t$  denote the induced subgraph $G[Y_t]$. 
A node $t\neq r$ in a rooted tree-cut decomposition is \emph{thin} if $\adh(t)\leq 2$ and \emph{bold} otherwise.


%

	A tree-cut decomposition $(T,\mathcal{X})$ is \emph{nice} if it satisfies the following condition for every thin node $t\in V(T)$: $N(Y_t)\cap (\bigcup_{b\text{ is a sibling of }t}Y_b)=\emptyset$.
	The intuition behind nice tree-cut decompositions is that we restrict the neighborhood of thin nodes in a way which facilitates dynamic programming. Every tree-cut decomposition can be transformed into a nice tree-cut decomposition of the same width in cubic time~\cite{Ganian0S15}.
	
	For a node $t$, we let $B_t=\{ b\text{ is a child of }t | |N(Y_b)|\leq 2\wedge N(Y_b)\subseteq X_t \}$ denote the set of thin children of $t$ whose neighborhood is a subset of $X_t$, and we let $A_t= \{a\text{ is a child of }t | a\not \in B_t \}$ be the set of all other children of $t$.
	Then $|A_t|\leq 2k+1$ for every node $t$ in a nice tree-cut decomposition~\cite{Ganian0S15}.
	\fi

%
	
%

We refer to previous work~\cite{MarxWollan14,Wollan15,KimOPST18,Ganian0S15} for a more detailed comparison of tree-cut width to other parameters. Here, we mention only that tree-cut width is dominated by treewidth and dominates treewidth plus maximum degree, which we denote $\degtw(G)$.
	\begin{lemma}[\hspace{-0.001cm}\cite{Ganian0S15,MarxWollan14,Wollan15}]
		\label{lem:comparison}
		For every graph $G$, $\tw(G)\leq 2\tcw(G)^2+3\tcw(G)$ and $\tcw(G)\leq 4\degtw(G)^2$.
	\end{lemma}
	
\section{Edge-Cut Width}
\label{sub:lfen} 

Let us begin by considering a maximal spanning forest $T$ of a graph $G$, and recall that $E(G)-T$ forms a minimum feedback edge set in $G$; the size of this set is commonly called the \emph{feedback edge number}~\cite{BentertHHKN20,GanianO21,GolovachKKL22}, and it does not depend on the choice of $T$. We will define our parameter as the maximum number of edges from the feedback edge set that form cycles containing some particular vertex $v\in V(G)$.

Formally, for a graph $G$ and a maximal spanning forest $T$ of $G$, let the \emph{local feedback edge set} at $v\in V$ be 
$E_{\loc}^{G,T}(v)=\{uw\in E(G)\setminus E(T)~|~$ the unique path between $u$ and $w$ in $T$ contains $v\}$; we remark that this unique path forms a so-called \emph{fundamental cycle} with the edge $uw$.
The \emph{\width} of $(G,T)$ (denoted $\widthshort(G,T)$) is then equal to $1+\max_{v\in V} |E_{\loc}^{G,T}(v)|$, and the \emph{\width} of $G$ is the smallest \width\, among all possible maximal spanning forests of $G$.

Notice that the definition increments the \width\ of $T$ by $1$. This ``cosmetic'' change may seem arbitrary, but it matches the situation for treewidth (where the width is the bag size minus one) and allows trees to have a width of $1$. Moreover, defining \width\ in this way provides a more concise description of the running times for our algorithms, where the records will usually depend on a set that is one larger than $|E_{\loc}^{G,T}(v)|$. We note that the predecessor to \width, called the \emph{local feedback edge number}~\cite{GanianKorchemna21}, was defined without this cosmetic change and hence is 
equal to \width\ minus one.

While it is obvious that $\widthshort(G)$ is upper-bounded by (and hence dominates) the feedback edge number of $G$ ($\fen(G)$), we observe that graphs of constant $\widthshort(G)$ can have unbounded feedback edge number---see Figure~\ref{fig:lfen}. 
We also note that Ganian and Korchemna established that \width\ is dominated by tree-cut width. 

 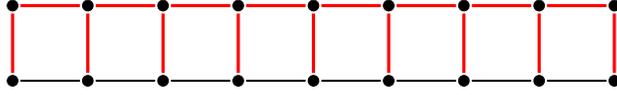
\begin{figure}
 \begin{center}
			\begin{tikzpicture}
				
				\vertex (n1) at (1,1) {};
				\vertex (n2) at (2,1) {};
				\vertex (n3) at (3,1) {};
				\vertex (n4) at (4,1) {};
				\vertex (n5) at (5,1) {};
                                           \vertex (n6) at (6,1) {};
                                           \vertex (n7) at (7,1) {};
                                           \vertex (n8) at (8,1) {};
                                           \vertex (n9) at (9,1) {}; 

                                           \vertex (m1) at (1,2) {};
				\vertex (m2) at (2,2) {};
				\vertex (m3) at (3,2) {};
				\vertex (m4) at (4,2) {};
				\vertex (m5) at (5,2) {};
                                           \vertex (m6) at (6,2) {};
                                           \vertex (m7) at (7,2) {};
                                           \vertex (m8) at (8,2) {};
                                           \vertex (m9) at (9,2) {}; 

                                           \draw [thick](n1)--(n2);
                                           \draw [thick](n2)--(n3);
                                           \draw [thick](n3)--(n4);
                                           \draw [thick](n4)--(n5);
                                           \draw [thick](n5)--(n6);
                                           \draw [thick](n6)--(n7);
                                           \draw [thick](n7)--(n8);
                                           \draw [thick](n8)--(n9);
 
                                           \draw[red, very thick] (m1)--(m2);
                                           \draw[red, very thick] (m2)--(m3);
                                           \draw[red, very thick] (m3)--(m4);
                                           \draw[red, very thick] (m4)--(m5);
                                           \draw[red, very thick] (m5)--(m6);
                                           \draw[red, very thick] (m6)--(m7);
                                           \draw[red, very thick] (m7)--(m8);
                                           \draw[red, very thick] (m8)--(m9);
                                    
				                           \draw[red, very thick] (n1)--(m1);
				\draw[red, very thick] (n2)--(m2);
                                           \draw[red, very thick] (n3)--(m3);
				\draw[red, very thick] (n4)--(m4);
                                           \draw[red, very thick] (n5)--(m5);
				\draw[red, very thick] (n6)--(m6);
                                           \draw[red, very thick] (n7)--(m7);
				\draw[red, very thick] (n8)--(m8);
                                           \draw[red, very thick] (n9)--(m9);
			\end{tikzpicture}
			\vspace{-0.5cm}
			\end{center}
\caption {Example of a graph $G$ with a spanning tree $T$ (marked in red) such that $\widthshort(G)=\widthshort(G,T)=3$. The feedback edge number of $G$, i.e., its edge deletion distance to acyclicity, is exactly the number of black edges and can be made arbitrarily large in this fashion while preserving $\widthshort(G)=3$.\label{fig:lfen}}
\end{figure}

\iflong
\begin{proposition}[\hspace{-0.001cm}\cite{GanianKorchemna21}]
\fi
\ifshort
\begin{proposition}[$\star$, \cite{GanianKorchemna21}]
\fi
\label{pro:lfescompare}
For every graph $G$, $\tcw(G)\leq \widthshort(G) \leq \fen(G)+1$.
\end{proposition}

\iflong
\begin{proof}
Let us begin with the second inequality. Consider an arbitrary spanning tree $T$ of~$G$. 
Then for every $v \in V(G)$, $E_{\loc}^T(v)$ is a subset of a feedback edge set corresponding to the spanning tree $T$, so $|E_{\loc}^T(v)| \leq \fen(G)$ and the claim follows.

To establish the first inequality, we will use the notation and definition of tree-cut width from previous work~\cite[Subsection 2.4]{GanianKO21}.
Let $T$ be the spanning tree of $G$ with $\widthshort(G,T)=\widthshort(G)$. We construct a tree-cut decomposition $(T, \XXX)$ where each bag contains precisely one vertex, notably by setting $X_t=\{t\}$ for each $t\in V(T)$.  Fix any node $t$ in $T$ other than root, let $u$ be the parent of $t$ in $T$. All the edges in $G\setminus ut$ with one endpoint in the rooted subtree $T_t$ and another outside of $T_t$ belong to $E^T_{loc}(t)$, so $\adh_T(t)=|\cut(t)|\leq |E^T_{loc}(t)| \leq \widthshort(G)-1$.\\\\ Let $H_t$ be the torso of $(T, \XXX)$ in $t$, then $V(H_t)=\{t,z_1...z_l\}$ where $z_i$ correspond to connected components of $T \setminus t$, $i\in [l]$. In $\tilde H(t)$, only $z_i$ with degree at least $3$ are preserved. 
But all such $z_i$ are the endpoints of at least two edges in $|E^T_{loc}(t)|$, so $\tor(t)=|V(\tilde H_t)|\le 1+ |E^T_{loc}(t)| \le \widthshort(G)$. Thus  $\tcw(G)\leq \widthshort(G)$. \end{proof}
\fi

As for the converse, we already have conditional evidence that \width\ cannot dominate tree-cut width: \textsc{Bayesian Network Structure Learning} is \W{1}-hard w.r.t.\ the latter, but fixed-parameter tractable w.r.t.\ the former~\cite{GanianKorchemna21}. We conclude our comparisons with a construction that not only establishes this relationship unconditionally, but---more surprisingly---implies that \width\ is incomparable to $\degtw$.

\iflong
\begin{lemma}
\fi
\ifshort
\begin{lemma}[$\star$]
\fi
\label{lem:laddertree}
For each $m\in \Nat$, there exists a graph $G_m$ of degree at most $3$, tree-cut width at most $2$, and \width\ at least $m+1$.
\end{lemma}

\iflong
\begin{proof}
\fi
\ifshort
\begin{proof}[Sketch]
\fi
We start from two regular binary trees $Y$ and $Y'$ of depth $m$, i.e., rooted binary trees where every node except leaves has precisely two children and the path from any leaf to the root contains $m$ edges.
We glue $Y$ and $Y'$ together by identifying each leaf of~$Y$ with a unique leaf of $Y'$ (see the left part of Figure \ref{fig: ecw_neq_tcw2} for an illustration). It remains to show that the resulting graph, which we denote $G_m$, has the desired properties. 
\begin{figure}[htb]
\begin{center}
\vspace{-0.3cm}
\includegraphics[width=0.75\textwidth]{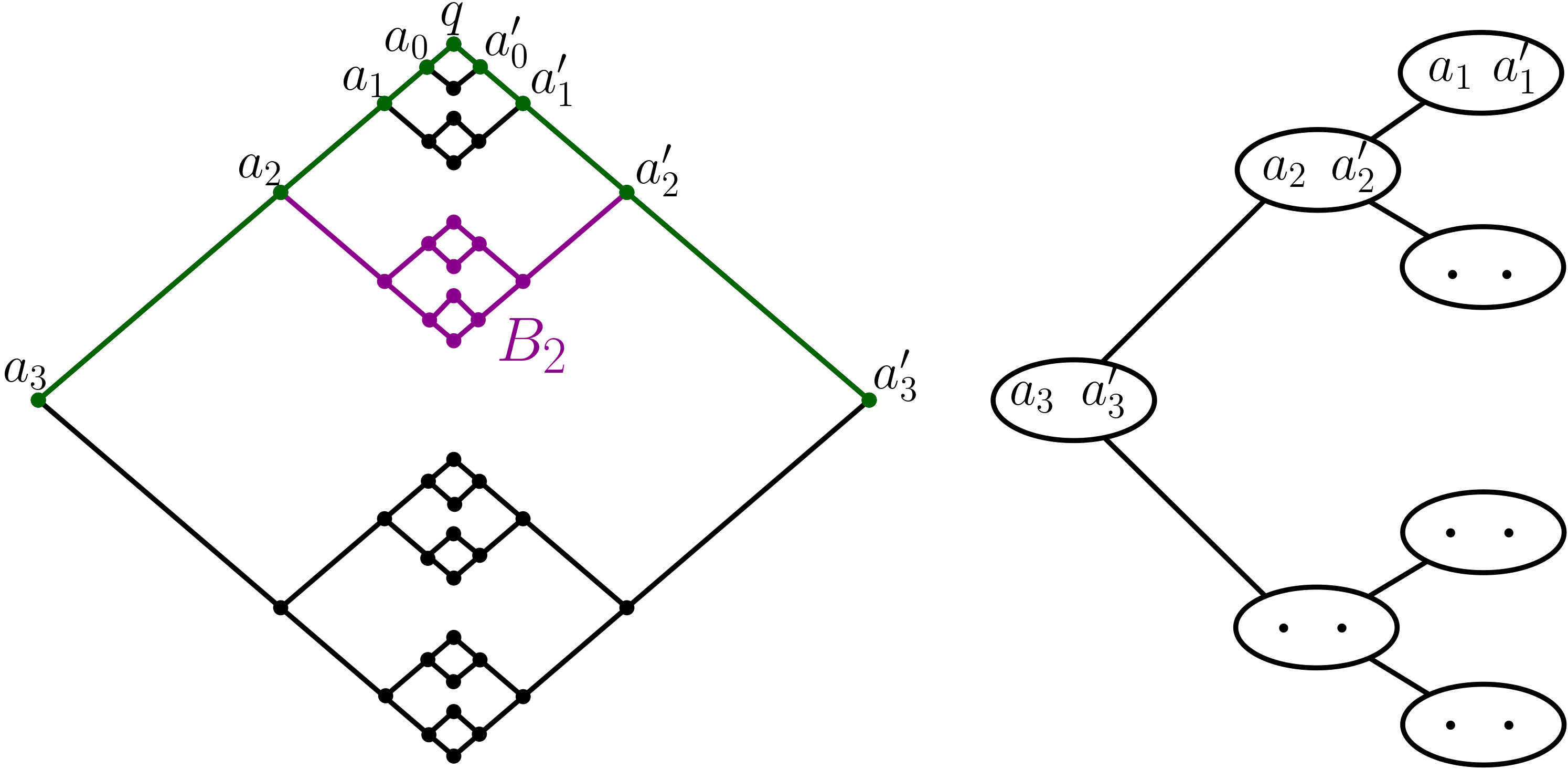}\vspace{-0.6cm}
\end{center}
\caption{\textbf{Left}: Graph $G_4$, where the roots of $Y$ and $Y'$ are $a_3$ and $a_3'$, the path $\pi$ is green and $B_2$ is violet. \textbf{Right}: Fragment of the tree-cut decomposition $(Y^*, \chi)$ of $G_4$.}
\label{fig: ecw_neq_tcw2} 
\end{figure}
\iflong

Consider arbitrary spanning tree $T$ of $G_m$. There exists a unique path $\pi \subseteq T$ between the roots $r$ and $r'$ of $Y$ and $Y'$. Observe that $G_m - \pi$ is a disjoint union of $m$ graphs $G_l$, $l \in [m-1]$. We add to every such $G_l$ two edges which connect it with $\pi$ and denote the resulting graph by $B_l$. Then every $B_l$ contains at least one edge that contributes to the local feedback edge set of $q\in V(\pi)$, where $q$ is a leaf in $Y$ and $Y'$. Indeed, fix any $l\in [m-1]$ and denote by $a_l$ and $a_l'$ the vertices of $B_l$ intersecting $\pi$ in $Y$ and $Y'$ correspondingly. As $T$ is a tree, $T - {q}$ is a union of two trees: one containing $a_l$ and another containing $a_l'$. Hence every vertex ob $B_l$ is connected to precisely one of $a_l$ and~$a_l'$ in $T - {q}$. In particular, there exists an edge $e_l$ of $B_l$ such that one endpoint of $e_l$ is connected to $a_l$ and another is connected to $a_l'$ in $T - {q}$. Then $e_l$ belongs to the local feedback edge set of every vertex of $\pi$ that lies between $a_l$ and $a_l'$, in particular, to the local feedback edge set of $q$. As $B_l$ and $B_{l'}$ don't share edges for any $l\ne l'$, this results in $|E_{\loc}^{G_m,T}(q)|\ge m$. Since the inequality holds for any choice of $T$, we may conclude that $\widthshort(G_m)\ge m+1$. 

To compute the tree-cut width of $G_m$ has, consider its tree-cut decomposition $(Y^*, \chi)$ where $Y^*$ is a regular binary tree of depth $m$ and $\chi$ is defined as follows. Let $h:V(Y^*) \to V(Y)$ and $h':V(Y^*) \to V(Y')$ be bijections such that (1) if $y$ is a leaf of~$Y^*$ then $h(y)$ and $h'(y)$ are identified leaves of $Y$ and $Y'$, and (2) if $y_1$ is a parent of $y_2$ in $Y^*$ then $h(y_1)$ is a parent of $h(y_2)$ in $Y$ and $h'(y_1)$ is a parent of $h'(y_2)$ in $Y'$. Further, for every node $y$ of $Y^*$ we define its bag to be $X_{y}=\{h(y), h'(y)\}$ (see the right part of Figure \ref{fig: ecw_neq_tcw2} for the illustration). Observe that the adhesion of every node as well as size of each bag is at most $2$, and all the children are thin, therefore, $\tcw(G_m) = 2$. \fi
\qed \end{proof}

%
%

Since it is known that treewidth dominates tree-cut width (see Lemma~\ref{lem:comparison}), Lemma~\ref{lem:laddertree} implies that \width\ does not dominate $\degtw$. Conversely, it is easy to build graphs with unbounded $\degtw$ and bounded \width\ (e.g., consider the class of stars). Hence, we obtain that \width\ is incomparable to $\degtw$. An illustration of the parameter hierarchy including \width\ is provided in Figure~\ref{fig: hierarchy}.
\begin{figure}[htb]
\begin{minipage}[c]{0.4\textwidth}
\begin{center}
\includegraphics[width=0.4 \textwidth]{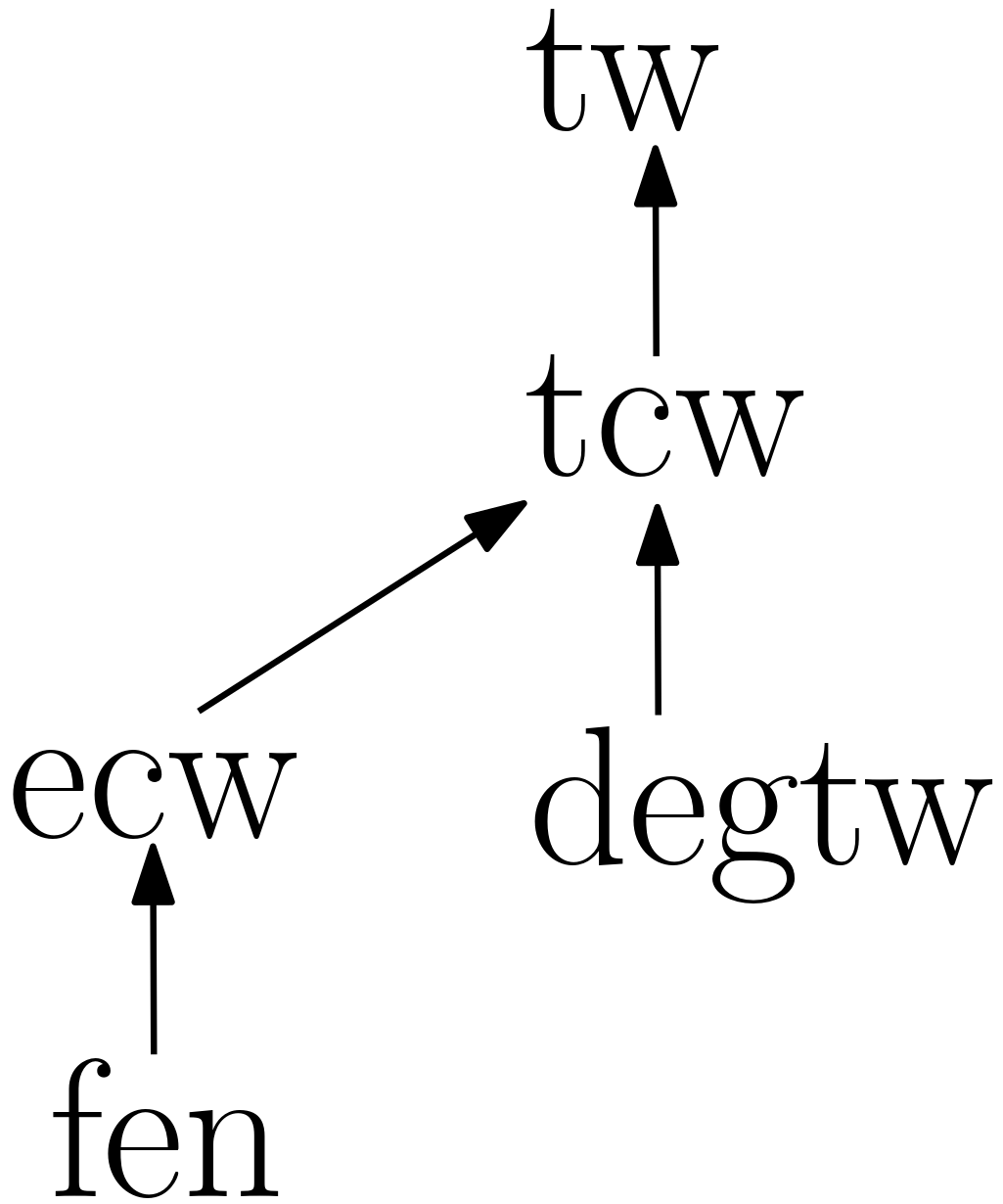}\vspace{-0.3cm}
\end{center}
\end{minipage}
\hfill
\begin{minipage}[c]{0.58\textwidth}
\vspace{-0.3cm}
\caption{Position of edge-cut width in the hierarchy of graph parameters. Here an arrow from parameter $\beta$ to parameter $\alpha$ represents the fact that $\alpha$ dominates $\beta$, i.e., there exists a function $p$ such that for each graph $G$, $\alpha(G)\leq p(\beta(G))$. We use $\fen$ to denote the feedback edge number.} 
\vspace{-0.3cm}
\end{minipage}
\label{fig: hierarchy} 
\end{figure}

Next, we note that even though Lemma~\ref{lem:comparison} and Proposition~\ref{pro:lfescompare} together imply that $\tw(G) \leq 2\ecw(G)^2 + 3\ecw(G)$, one can in fact show that the gap is linear. This will also allow us to provide a better running time bound in Section~\ref{sec:compute}.

\iflong
\begin{lemma}
\fi
\ifshort
\begin{lemma}[$\star$]
\fi
\label{lem:ecwtw}
		For every graph $G$, $\tw(G)\leq \widthshort(G)$.
\end{lemma}

\iflong
\begin{proof}
Let $T$ be the spanning tree of $G$ such that $\widthshort(G)=\widthshort(G, T)$. We arbitrarily pick a root $r$ in $T$ and construct the tree decomposition $(T, \{\beta_v\}_{v \in V(T)})$ of $G$ as follows. At first, for every $v\in V(G)$, we add to $\beta_v$ the vertex $v$ and the parent of $v$ in $T$ (if it exists). Obviously, after this step each vertex $v$ of $G$ appears in some bag and every edge of $T$ is contained as a subset in some bag. Moreover, $v$ appears only in $\beta_v$ and in the bags of children of $v$ in $T$, which results in a connected subtree of $T$. 

To complete the construction, we process feedback edges one by one. For every $e\in E(G) \setminus E(T)$, we arbitrarily choose an endpoint $u$ of $e=uw$ and add $u$ to each bag $\beta_v$ such that $u\in E_{\loc}^{G,T}(v)$. Note that any such step does not violate the connectivity condition. Indeed, we add $u$ to the bags of all vertices which lie on the path between the endpoints of $e$ in $T$. In particular, the path hits $u$ whose bag $\beta_u$ initially contained $u$. Finally, both endpoints of $e$ appear in $\beta_w$.
In the resulting decomposition, for each $v\in V(G)$ it holds that $|\beta_v| \le 2+ E_{\loc}^{G,T}(v) \le 1+\ecw(G)$. Hence the width of $(T, \{\beta_v\}_{v \in V(T)})$ is at most $\ecw(G)$. \qed
\end{proof}
\fi

Last but not least, we show that---also somewhat surprisingly--- \width\ is not closed under edge or vertex deletion. \ifshort (see Figure~\ref{fig: ecw_neq_extecw}, $\star$)\fi

\iflong
For the edge-deletion case, we refer readers to Figure~\ref{fig: ecw_neq_extecw} which illustrates a graph $G$ along with a spanning tree witnessing that $\widthshort(G)\le 4$. On the other hand, any spanning tree $T$ of $G-{ac}$ must contain both edges $ab_i$ and $b_ic$ for some $i\in\{1,2,3\}$. We will assume that those edges are $ab_1$ and $b_1c$, since the other cases are symmetrical. Then $T$ contains precisely one edge of each pair $(ab_2,b_2c)$ and $(ab_3,b_3c)$. The other, ``missing'' edge from each pair contributes to the local feedback edge set of $b_1$. Together with two missing edges of 3-cycles that intersect $b_1$, this results in $|E_{\loc}^{G-{ac},T}(b_1)|\ge 4$ and, since similar situation happens for any choice of a spanning tree, we conclude that $\widthshort(G-{ac})\ge 5$. The vertex deletion case can be argued analogously using the graph obtained from $G$ by subdividing the edge $ac$.
\fi

\begin{corollary}
\label{cor:deletion}
There exist graphs $G$ and $H$ such that $\widthshort(G-e)>\widthshort(G)$ and $\widthshort(H-v)>\widthshort(H)$ for some $e\in E(G)$ and $v\in V(H)$.
\end{corollary}

\begin{figure}[htb]
\begin{center}
\vspace{-0.5cm}
\includegraphics[width=0.9\textwidth]{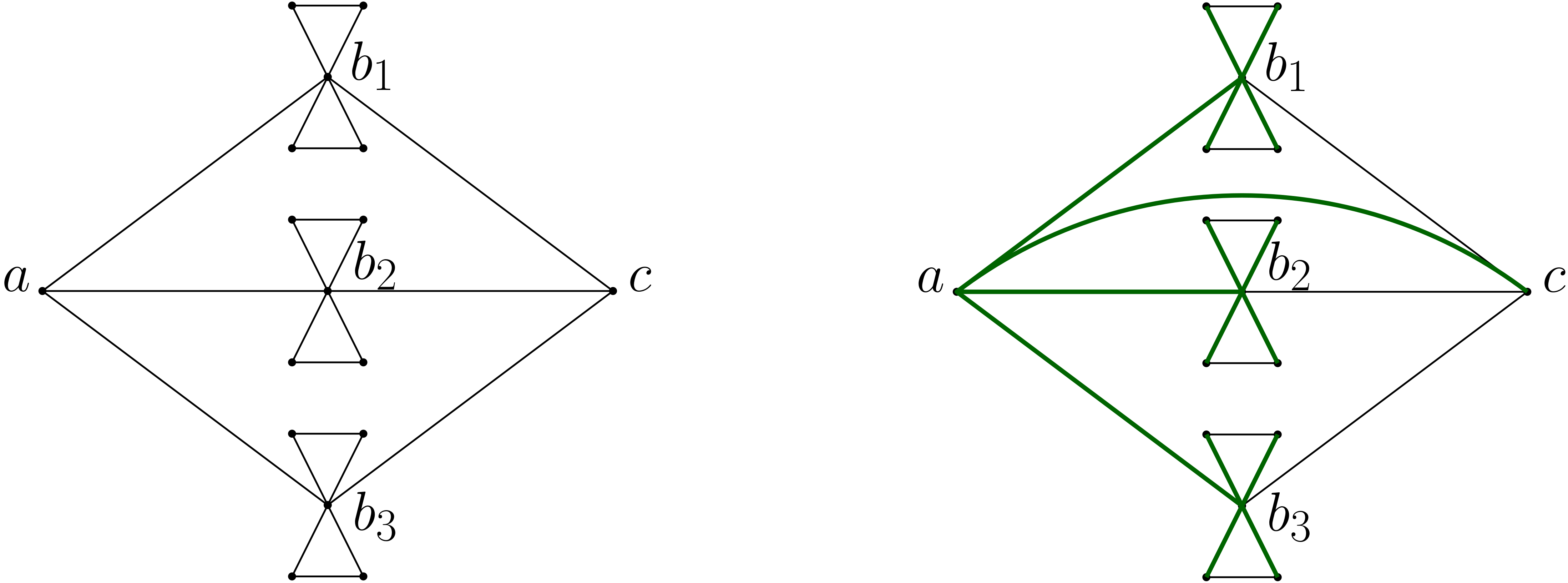}\vspace{-0.6cm}
\end{center}
\caption{\textbf{Left}: Graph $G-{ac}$ of $\widthshort(G-{ac})\ge 5$. \textbf{Right}: Green tree witnessing that $\widthshort(G)\le 4$.}
\label{fig: ecw_neq_extecw}
\end{figure}




\section{Computing Edge-Cut Width}
\label{sec:compute}
Before we proceed to the algorithmic applications of \width, we first consider the question of computing the parameter along with an optimal ``decomposition'' (i.e., spanning tree). Here, we provide an explicit fixed-parameter algorithm for this task.

By Lemma~\ref{lem:ecwtw},  the treewidth of $G$ can be linearly bounded by $\widthshort(G)$. The algorithm uses this to perform dynamic programming on a tree decomposition $(T,\{\beta_t\}_{t\in V(T)})$ of~$G$.
For a node $t \in V(T)$, we let $Y_t$ be the union of all bags $\beta_s$ such that $s$ is either $t$ itself or a descendant of $t$ in $T$, and let $G_t$ be the subgraph $G[Y_t]$ of $G$ induced by $Y_t$.

\iflong
\begin{lemma}
\fi
\ifshort
\begin{lemma}[$\star$]
\fi
\label{lem:ecw}
Given an $n$-vertex graph $G$ of treewidth $k$ and a bound $w$, it is possible to decide whether $G$ has edge-cut width at most $w$ in time $k^{\bigoh(w k^2)}\cdot n$. If the answer is positive, we can also output a spanning tree of $G$ of edge-cut width at most $w$.
\end{lemma}
Using the relation between treewidth and edge-cut width above, we immediately obtain:
\begin{theorem}
\label{thm:computeecw}
Given a graph $G$, the edge-cut width $\widthshort(G)$ can be computed time $\widthshort(G)^{\widthshort(G)^{3}}\cdot n$.
\end{theorem}

\iflong
\begin{proof}[of Lemma \ref{lem:ecw}]
\fi
\ifshort
\begin{proof}[Proof Sketch of Lemma \ref{lem:ecw}]
\fi
Without loss of generality, we assume that $G$ is connected. Using state-of-the-art approximation algorithms~\cite{BodlaenderDrangeDregiFominLokshtanov16,2approxtw}, we first compute a ``nice'' tree decomposition $(T,\{\beta_t\}_{t\in V(T)})$ with root $r \in V(T)$ of width $k = \bigoh(\tw(G))$ in time $2^{\bigoh(k)}\cdot n$.


On a high level, the algorithm relies on the fact that if $G$ has \width\ at most~$w$, then at each bag $\beta_t$ the number of unique paths contributing to the \width\ of vertices in $\beta_t$ is upper-bounded by $|\beta_t| w \leq kw$.
Otherwise, at least one of the vertices in $\beta_t$ would lie on more than $w$ cycles.
We can use this to branch on how these at most $kw$ edges are routed through the bag.


At each vertex $t \in T$ of the tree decomposition, we store \emph{records} that consist of:
\begin{itemize}[topsep=0pt]
\item  an acyclic subset $F$ of edges of $G[\beta_t]$, 
\item a partition $\mathcal{C}$ of $\beta_t$, and 
\item two multisets $\fut,\past$ of sequences of vertex-pairs $(u,v)$ from $\beta_t$, with the following property:
\begin{itemize}[topsep=0pt]
\item Every vertex of $\beta_t$ appears on at most $w$ distinct $u$-$v$ paths, where $(u,v)$ is a pair of vertices in a sequence in $\fut$ or $\past$.
\item $v_i$ and $u_{i+1}$ are not connected by an edge in $\beta_t$.
\end{itemize}
\end{itemize}

The semantics of these records are as follows:
For every spanning tree of width at most $w$, the record describes the intersection of the solution with $G[\beta_t]$, and the intersection of every fundamental cycle of this solution with $G[\beta_t]$.
We encode the path that a cycle takes through $G[\beta_t]$ via a sequence of vertex pairs that indicate where the path leaves and enters $G[\beta_t]$ from the outside (it may be that these are the same vertex).
More precisely, $\past$ contains those cycles that correspond to an edge that has already appeared in $G_t$, whereas $\fut$ corresponds to those cycles that correspond to an edge not in $G_t$.
In particular, this allows to reconstruct on how many cycles a vertex of~$\beta_t$ lies. 
The partition $\mathcal{C}$ says which vertices of $\beta_t$ are connected via the solution in $G_t$.

To be more precise, let $t \in T$ and let $S$ be an acyclic subset of edges of $G$ that has width at most $w$ on $G_t$ (that is, each vertex of $S$ lies on at most $w$ fundamental cycles of $S$ in $G_t$). We call such $S$ \emph{partial solutions} at $t$. Then, we let the \emph{$t$-projection} of $S$ be defined as $(F,\mathcal{C},\fut,\past)$, where 
\begin{itemize}[topsep=0pt]
\item $F = S \cap G[\beta_t]$.
\item $\mathcal{C}$ is a partition of $F$ according to the connected components of $S$ in $G_t$.
\item Let $C_e$ be a fundamental cycle of $S$ in $G$ corresponding to the edge $e \in G - S$. Then, there is a sequence $P_e  = ((u_1,v_1),\ldots,(u_t,v_t))$ in either $\fut$ or $\past$ of vertex pairs such that the intersection of $C_e$ with $S$ traverses $F$ along the unique $u_i$-$v_i$ paths in the order they appear in $P_e$ (note that $u_i = v_i$ is possible, in which case the path contains just the vertex $u_i$).
\item For each fundamental cycle $C_e$ of $S$ in $G$, if $e \in G_t$, then $P_e \in \past$, otherwise, $P_e \in \fut$.
\end{itemize}
Note that $P_e$ can (and often will) be the empty sequence $P_e = \emptyset$.
Moreover, we assume that the correspondence between $\fut \cup \past$ and the edges in $G-S$ is bijective, in the sense that if two edges $e,e'$ produce the same sequence $P_e=P_{e'}$, then $P_e$ and~$P_{e'}$ occur as two separate copies in $\fut \cup \past.$


The encoding length of a single record is $\bigoh(wk^2\log k)$, dominated by the at most $kw$ sequences $P_e$ of $k$ pairs of vertices each, with indices having $\bigoh(\log k)$ bits.
Overall, the number of records is hence bounded by $2^{\bigoh(w k^2 \log k)}$. 

For each $t\in T$, we store a set of records $\mathcal{R}(t)$ that has the property that $\mathcal{R}(t)$ contains the set of all $t$-projections of spanning trees of width at most $w$ (that is, projections of solutions of the original instance). In addition, we require for every record in $\mathcal{R}(t)$ that there is a partial solution $S$ of $G_t$ of width at most $w$ that agrees with $F,\mathcal{C}$ and $\past$ of the record. In this case, we call $\mathcal{R}(t)$ \emph{valid}.
Supposing correctness of this procedure, $G$ is a YES-instance if and only if $(F_r,\mathcal{C}_r,\past_r,\fut_r) \in \mathcal{R}(r)$, with $F_r = \mathcal{C}_r = \fut_r = \emptyset$, $\past_r = \{\emptyset^{m-n}\}$, and a NO-instance otherwise.

\ifshort
To conclude the proof, it now suffices to compute $\mathcal{R}(t)$  in a leaf-to-root fashion.
\fi
\iflong
We compute $\mathcal{R}(t)$ bottom-up along the nice tree-decomposition depending on the type of the node $t$ as follows:

\textbf{At a leaf-node,} per convention, $\beta_t = \emptyset$, and since $G_t$ is the empty graph, any spanning tree $S$ has width at most $w$ on $G_t$. This implies that any $t$-projection $(F,\mathcal{C},\past,$ $\fut)$ of such $S$ satisfies $F = \mathcal{C} = \past = \emptyset, \fut = \{\emptyset^{n-m}\}$.
It therefore suffices to set $\mathcal{R}(t) = \{(\emptyset,\emptyset,\emptyset,\{\emptyset^{n-m}\})\}$, and this is valid.

\textbf{At an introduce-node,} let the vertex introduced at $t$ be $v \in G$, and let $s$ be the unique child of $t$ in $T$. By definition, $\beta_t = \beta_s \cup \{v\}$.
We assume by inductive hypothesis that $\mathcal{R}(s)$ is valid.
Consider now any solution $S$ of width at most $w$ on $G_t$.
This solution will be of width at most $w$ also on $G_s$.
Hence, since $\mathcal{R}(s)$ is assumed valid, there is a record $(F_s,\mathcal{C}_s,\past_s,\fut_s)$ corresponding to the $s$-projection of $S$.

We first branch over the way that the edges incident with $v$ in $G[\beta_t]$ extend $F_s$.
Call this new set of edges $E_v$.
During this process, we discard any choice of $E_v$ that connects vertices within the same connected component as indicated by $\mathcal{C}_s$.

Furthermore, we discard any choice that implies cycles in the solution via $\fut$:
If there is an entry in $\fut_s$ that contains two consecutive pairs $(u,u'),(w,w')$ such that $u'$ and $w$ are now in the same component of $\mathcal{C}$ (that is, were connected by adding $v$ to $G_s$), and one of $u'$ or $w$ is not a neighbor of $v$, then this would imply two $u'$-$w$ paths: $u'$ and $w$, but not any of the vertices on the paths $u'$-$v$ and $v$-$w$ lie on the fundamental cycle corresponding to the entry in $\fut_s$ containing $(u,u'),(w,w')$, yielding two paths: One through the cycle, the other through $v$ via $E_v$. Therefore, this choice of $E_v$ can be discarded.

Then, for every edge $(v,u)$ incident with $v$ that was not chosen into $E_v$, there must be a sequence of pairs $P$ in $\fut_s$ such that the last vertex in the last pair of the sequence $P$ is $u$, otherwise we may discard $E_v$ (since the corresponding fundamental cycle wasn't reflected in $\fut_s$.)
We branch over all ways of choosing $P_1,\ldots,P_d$ for each edge $e_1,\ldots,e_d$ incident to $v$ that is not in $E_v$.
For each $i = 1,\ldots,d$, if $P = P_i$ just consists of the single pair $(u,u)$, we add the single pair $(v,v)$ to $P$, and move $P$ to $\past$ (since the feedback edge $(v,u)$ is now part of $G_t$).
Otherwise, if the first pair $(w,w')$ in $P$ is distinct from $(u,u)$, we add the pair $(v,w')$ to $P$, remove $(w,w')$ from $P$, and add $P$ to $\past$.

We now update $\past$ and $\fut$ as follows: 
If there is a consecutive pair $(u,u'),(w,w')$ in an element of $\past_s$ or $\fut_s$ such that $u'$ and $w$ are neighbors of $v$, replace the subsequence $(u,u'),(w,w')$ by $(u,w')$: any other choice of connecting $u'$ and $w$ through a path than directly via $v$ would imply a cycle. 
In any case, add the resulting sequence to $\past$ or $\fut$, respectively.

We then branch over the choices of extending fundamental cycles along $v$: 
For each pair in a sequence in $\past$ or $\fut$ that contains a neighbor $u$ of $v$ connected via $E_v$, branch over whether or not to route this fundamental cycle via $v$ by replacing $(u,w)$ by $(v,w)$ or $(w,u)$ by $(w,v)$, respectively. 

If during any of the choices for $E_v,P_1,\ldots,P_d$ and the extensions of the fundamental cycles via $v$, the solution would have to route more than $w$ cycles over any vertex of $\beta_t$ (as can be checked by tracing out all the pairs in the sequences now contained in $\fut$ and $\past$), discard the choice.
If there is no way to choose the above without exceeding the width bound, discard the entire choice of record and consider the next record in $\mathcal{R}(s)$.

If this is not the case, then, for a choice of $E_v$ (i.e., how to extend $F_s$), $P_1,\ldots,P_d$ (i.e., how to route the new edges in $G_t$ in $\past$) and a choice of extending the existing cycles in $\past_s$ and $\fut_s$ to in-- or exclude $v$,
we branch over how many additional fundamental cycles $v$ outside of $G_t$ will be part of,
and add as many copies of the sequence consisting just of $(v,v)$ to $\fut$, simultaneously decreasing the multiplicity of $\emptyset$ in $\fut_s$ by as many, and adding the result to $\fut$.

Finally, add $(F_s \cup E_v, \mathcal{C},\past,\fut)$ to $\mathcal{R}(t)$, and consider the next entry of $\mathcal{R}(s)$.
Since any partial solution of width at most $w$ on $G_t$ will have to extend its $s$-projection in one of the above ways, this generates all possible $t$-projections (and possibly some additional records with the same $F,\mathcal{C},\past$). In particular, the generated set $\mathcal{R}(t)$ is valid.
This completes the description of the introduce step. 

The running time of this step is dominated by branching over the sequences $P_1,\ldots,P_d$.
Since $d\leq k$ and there are at most $kw$ sequences in total, we have $(kw)^k = 2^{\bigoh(wk\log k)}$ choices at most, for each of the $2^{\bigoh(wk^2\log k)}$ records in $\mathcal{R}(s)$, and processing each choice only adds a lower-order term in the running time. Therefore, this step takes time $2^{\bigoh(wk^2\log k)}$.

\textbf{At a forget-node}, let the vertex forgotten at $t$ be $v \in G$, and let $s$ be the unique child of $t$ in $T$. By definition, $\beta_t = \beta_s - \{v\}$.
We assume by inductive hypothesis that $\mathcal{R}(s)$ is valid,
and let $(F_s,\mathcal{C}_s,\past_s,\fut_s) \in \mathcal{R}(s)$.

If $\{v\} \in \mathcal{C}$ (that is, $v$ is a single component in the intersection of any solution that projects to the current record with $\beta_t$), then discard the choice for the record and consider the next element of $\mathcal{R}(s)$. In this case, the component that contains $v$ in any partial solution conforming with the record could never be completed to form a connected subgraph.

If $(v,v)$ appears as part of a sequence in $\fut_s$ or $\past_s$, remove $(v,v)$ from the sequence. 
If, on the other hand, $(v,u)$ is part of any sequence in $\past_s$ or $\fut_s$ for some $u\neq v$,
replace $(v,u)$ by $(v',u)$, where $v'$ is the next vertex on the unique $v$-$u$ path in $F_s$ (and $u=v'$ is possible).
In both cases, add the resulting sequence (which is possibly equal to the empty sequence) to $\fut$ or $\past$, respectively. If the empty sequence would be added to $\fut$, discard the current record (since there is no way of closing this fundamental cycle in the future that can involve $v$).

We remove all edges involving $v$ from $F_s$ to obtain $F$ and update $\mathcal{C}_s$ by removing $v$ from all sets it appears in, thereby obtaining $\mathcal{C}$. We add $(F,\mathcal{C},\past,\fut)$ to $\mathcal{R}(t)$. Since $G_t = G_s$, the set of solutions that contribute to the set of $t$-projections and $s$-projections doesn't change; we hence only have to update the $s$-projections to become $t$-projections, as we did, in order to obtain a valid set $\mathcal{R}(t)$.

The running time of this step is dominated by the running time at the introduce-nodes.

\textbf{At a join-node}, let $s$ and $s'$ be the two children of $t$ in $T$.
We consider all pairs of records in $\mathcal{R}(s)$ and $\mathcal{R}(s')$.
If $F_s \neq F_{s'}$ or $\fut_s \neq \fut_{s'}$, we discard the current choice.
Consider the transitive closures of the reachability relations on $\beta_t$ as induced by $\mathcal{C}_s$ and $\mathcal{C}_{s'}$, respectively. If their union (as multigraphs) produces a cycle (which could be two parallel edges $(u,v)$ and $(u,v)$ for some $u,v \in \beta_t = \beta_s = \beta_{s'}$), any solution that $s$-projected and $s'$-projected to $\mathcal{C}_{s}$ and $\mathcal{C}_{s'}$, respectively, would be cyclic on $G_t$. Hence, we may discard this choice of records.

If none of the above happens, we set $\past = \past_s \cup \past_{s'}$ as multisets, and check if this results in any of the vertices of $\beta_t$ coming to lie on more than $w$ fundamental cycles. If this is the case, we discard the current choice of records. If not, let $\mathcal{C}$ be finest common coarsening of the partitions $\mathcal{C}_s$ and $\mathcal{C}_{s'}$ (that is, the result of merging any two components that share a vertex, and exhausting this process). We let $F=F_s (=F_{s'}), \fut = \fut_s (=\fut_{S'})$, and set $\mathcal{R}(t) = (F,\mathcal{C},\past,\fut)$.

By a similar token as in the previous cases, this produces all possible $t$-projections of solutions of $G_s$ and $G_{s'}$ that are also solutions for $G_t$ of width at most $w$, and hence a valid set $\mathcal{R}(t)$.

Since we have to consider pairs of records that differ in $\past$, and $\past$ dominates the size of the records, the running time at the join-nodes dominates the running time at the introduce-nodes, and is bounded by $2^{\bigoh(wk^2\log k)}$.

Overall, the running time of the algorithm is bounded by $2^{\bigoh(wk^2 \log k)}\cdot n$.
By keeping one representative of a partial solution of $G_t$ per record at each node $t$ that $t$-projects to the current record, we can successively build a solution of width at most $w$.
\fi
\qed \end{proof}

\section{Algorithmic Applications of Edge-Cut Width}
\label{sec:problems}

Here we obtain algorithms for the following five \NP-hard problems (where a sixth problem mentioned in the introduction, \textsc{Precoloring Extension}, is a special case of \textsc{List Coloring}, and the fixed-parameter tractability of \textsc{Bayesian Network Structure Learning} and \textsc{Polytree Learning} follows from previous work~\cite{GanianKorchemna21}). In all of these, we will parameterize either by the \width\ of the input graph or of a suitable graph representation of the input. Recall that all problems are known to be \W{1}-hard when parameterized by tree-cut width~\cite{Ganian0S15,GozupekOPSS17,BredereckHKN19,GanianO21}, and here we will show they are all fixed-parameter tractable w.r.t.\ \width.

As a unified starting point for all algorithms, we will apply Theorem~\ref{thm:computeecw} to compute a minimum-width spanning tree $T$ of the input graph (or the graph representation of the input) $G$; the running time of Theorem~\ref{thm:computeecw} is also an upper-bound for the running time of all algorithms except for \textsc{MaxSRTI}, which has a quadratic dependence on the input size. Let $r$ be an arbitrarily chosen root in $T$. For each node $v\in V(T)$, we will use $T_v$ to denote the subtree of $T$ rooted at $v$. 
Without loss of generality, in all our problems we will assume that $G$ is connected.

The central notion used in our dynamic programming framework is that of a \emph{boundary}, which fills a similar role as the bags in tree decompositions. Intuitively, the boundary contains all the edges which leave $T_v$ (including the vertices incident to these edges).

\begin{definition}
For each $v\in V(T)$, the \emph{boundary} $\partial(v)$ of $T_v$ is the edge-induced subgraph of $G$ induced by those edges which have precisely one endpoint in $T_v$. 
\end{definition}
		
\newcommand{\forgottenv}{\mathcal{Y}}
\newcommand{\forgottenG}{\mathcal{G}}

Observe that for each $v\in V(T)$, $|E(\partial(v))|\leq \widthshort(G)$ and $|V(\partial(v))|\leq 2\widthshort(G)$. It will also sometimes be useful to speak of the graph induced by the vertices that are ``below'' $v$ in $T$, and so we set  $\forgottenv_v=\{w~|~w$ is a descendant of $v$ in $T\}$ and $\forgottenG_v=G[\forgottenv_v]$; we note that $v\in \forgottenv_v$. Observe that $\partial(v)$ acts as a separator between vertices outside of~$\forgottenv_v\cup V(\partial(v))$ and vertices in $\forgottenv_v\setminus V(\partial(v))$ 

%
%
%
%
%
%

\iflong
\subsection{Edge Disjoint Paths}
\fi
\ifshort
\smallskip
\noindent \textbf{Edge Disjoint Paths.} \quad
\fi
We start with the classical \textsc{Edge Disjoint Paths} problem, which has been extensively studied in the literature. While its natural counterpart, the \textsc{Vertex Disjoint Paths} problem, is fixed-parameter tractable when parameterized by treewidth, \textsc{Edge Disjoint Paths} is \W{1}-hard not only when parameterized by tree-cut width~\cite{GanianO21} but also by the vertex cover number~\cite{FleszarMS18}.

\pbDef{\textsc{Edge Disjoint Paths} (EDP)}
{A graph $G$ and a set $P$ of \emph{terminal pairs}, i.e., a set of subsets of $V(G)$ of size two.}
{Is there a set of pairwise edge disjoint paths connecting every set of terminal pairs in $P$?}

A vertex which occurs in a terminal pair is called a \emph{terminal} and a set of pairwise edge disjoint paths connecting every set of terminal pairs in $P$ is called a \emph{solution}. 

\iflong
\begin{theorem}\label{thm:EDP}
\fi
\ifshort
\begin{theorem}[$\star$]
\fi
\EDP\ is fixed-parameter tractable when parameterized by the \width\ of the input graph. 
\end{theorem}

\newcommand{\Rec}{\text{Record}}
\newcommand{\Recs}{\mathcal{R}}

\iflong
\begin{proof}
\fi
\ifshort
\begin{proof}[Sketch]
\fi
We start by defining the syntax of the records we will use in our dynamic program. For $v\in V(G)$, let a record be a tuple of the form $(S,D,R)$, where:

	\begin{itemize}[topsep=0pt]
		\item $S=\{(t_0,e_0),\dots,(t_i,e_i)\}$ where for each $j\in [i]$, $t_j \in \forgottenv_v$ is a terminal whose counterpart is not in $\forgottenG_v$, $e_j\in E(\partial(v))$, and where each terminal without a partner in $\forgottenv_v$ appears in exactly one pair,
		
		\item $D,R$ are sets of unordered pairs of elements from $E(\partial(v))$, and
		
		
		\item each edge of $E(\partial(v))$ may only appear in at most one tuple over all of these sets.
	\end{itemize}
	We refer to the edges in $S,D,R$ as \emph{single, donated} and \emph{received} edges, respectively, in accordance with how they will be used in the algorithm.
	Let $\Recs(v)$ be a set of records for~$v$. From the syntax, it follows that $|\Recs(v)|\leq 2^{\bigoh(k\log k)}$ for each $v\in V(G)$. 
	
	Let $P_v\subseteq (\forgottenv_v \cup V(\partial(v)))^2$ be a set that can be obtained from $P$ by the following three operations:
	\begin{itemize}[topsep=0pt]
	\item for some $\{a,b\}\in P$ where $a\in \forgottenv_v$, $b\not \in \forgottenv_v$, replacing $b$ by some $c\in V(\partial(v))$, and
	\item for some $a',b'\in V(\partial(v))\setminus\forgottenv_v$, adding $\{a',b'\}$ to $P_v$, and 
	\item for each $\{a,b\}\in P$ where $a,b \not \in \forgottenv_v$, remove $\{a,b\}$.
	\end{itemize}
	To define a partial solution we need the following graph $H_v$: \begin{itemize}[topsep=0pt]
		\item First, we add $\forgottenG_v \cup \partial(v)$ to $H_v$, where $\forgottenG_v\cup \partial(v)$ is the (non-disjoint) union of these two graphs.
		\item Next, we create for each edge $e \in E(\partial(v))$ a pendant vertex $v_e$ adjacent to the endpoint of $e$ that is outside of $\forgottenv_v$. Let $V_{\partial}$ denote the set of these new vertices.
		\item Finally, we add edges to $E(H_v)$ such that $V_{\partial}$ is a clique.
	\end{itemize}
	Let a partial solution at $v$ be a solution to the instance $(H_v,P_v)$ for some $P_v$ defined as above. 
	Obviously, since at the root $r$ we have that $\partial(r)$ is empty, $P_r=P$ and $H_r = G$.
	Notice that a partial solution at the root is a solution.
%
%
%
%
	 

	Consider then the set $\mathcal{W}$ containing all partial solutions at $v$. The $v$-\emph{projection} of a partial solution $W\in \mathcal{W}$ at $v$ is a record $(S_W,D_W,R_W)$ where:
	\begin{itemize}[topsep=0pt]
	\item $(t,e)\in S_W$ if and only if $t$ is a terminal in $\forgottenv_v$ whose counterpart $t'$ is not in $\forgottenv_v$ and~$e$ is the first edge in $E(\partial(v))$ encountered by the $t$-$t'$ path in $W$,
	\item $\{e_i,e_j\}\in D_W$ if and only if there is a path $Q\in W$ with $Q=e_i,e_{i+1},\dots,e_{j-1},e_j$ 
	such that the edges in $Q\setminus \{e_i, e_j\}$ are contained in $E(\forgottenG_v)$\footnote{Note that by the syntax, it follows that $e_i$ and $e_j$ are both contained in $\partial(v)$}, and 
%
	\item $\{e_i,e_j\}\in R_W$ if and only if there 
	is some $s$-$t$ path $Q\in W$ such that $s,t$ in $\forgottenv_v$, $e_i$ is the first edge in $E(\partial(v))$ that occurs in $Q$, and $e_j$ is the last edge in $E(\partial(v))$ that occurs in $Q$.
		
	\end{itemize}
	
	We say that $\Recs(v)$ is \emph{valid} if and only if it contains all $v$-projections of partial solutions in $\mathcal{W}$, and in addition, for every record in $\Recs(v)$, there is a partial solution such that its $v$-projection yields this record.
	
	Observe that if $\Recs(r)=\emptyset$, then $(G,P)$ is a NO-instance, while if $\Recs(r)=\{(\emptyset,\emptyset,\emptyset)\}$, then $R(r)$ is a YES-instance. To complete the proof, it now suffices to dynamically compute a set of valid records in a leaf-to-root fashion along $T$.	We note that if at any stage we obtain that a vertex $v$ has no records (i.e., $\Recs(v)=\emptyset$), we immediately reject.
	\iflong
%
	
	\textbf{If $v$ is a leaf}, we branch over all possible valid records by setting $R = \emptyset$ and letting~$D$ vary over all subsets of $\{\{e_1,e_2\}~|~e_1,e_2 \in E(\partial(v)) \}$. 
	In the case that $v$ is a terminal, we additionally let $S$ vary over all subsets of $\{(v,e)~|~ e \in E(\partial(v)) \}$. We discard choices of $S$ and $D$ where the same edge appears more than once over both sets.
	
%
%
	The set $\Recs(v)$ is trivially valid.
	
	\textbf{If $v$ is an internal node,} we proceed in the following way:
	First, we bound the number of children of $v$ by our parameter $k$. Then we branch over all possible combinations of records for the remaining children of $v$ to obtain $\Recs(v)$.

	We reduce the size of the subtree in the following way: Let $u$ be a child of $v$ with $E(\partial(u))=\{\{u,v\}\}$, i.e., $T_u$ has no edge that increases the size of the edge-cut width of $v$.
	\begin{itemize}[topsep=0pt]
		\item If there is no terminal pair with precisely one vertex in $T_u$, then delete $T_u$ along with all terminal pairs with both endpoints in $T_u$. 		
		\item If there is a single terminal pair $\{s,t\}$ with precisely one vertex, say $t$, in $T_u$, then replace $t$ with $v$ and delete $T_u$ along with all terminal pairs with both endpoints in $T_u$. (We remark that $v$ can be contained in multiple terminal pairs at the same time.)
		
		\item Otherwise, we correctly identify that this is a NO-instance.
	\end{itemize}
	
	Since $\forgottenG_u$ is connected to the remaining graph by a single edge it can connect only one terminal in $\forgottenv_u$ with a terminal in $V\setminus \forgottenv_u$. 	
	After this step there are at most $2(k-1)$ children left because at most $2(k-1)$ subtrees rooted at a child of $v$ can contribute to the \width\ of $v$.
	
	Let $u_1, \ldots, u_\ell$ with $\ell \leq 2(k-1)$ denote the remaining children of $v$. 
	First, we compute a set $\overline{\Recs(v)}$, in the same way we would compute $\Recs(v)$ if $v$ was a leaf. Our goal is to compute $\Recs(v)$ using the local set $\overline{\Recs(v)}$ and the partial results $\Recs(u_1), \ldots, \Recs(u_\ell)$.
	
	In the following we take one record each out of $\overline{\Recs(v)}, \Recs(u_1), \ldots, \Recs(u_\ell)$ and repeat the following process for each combination of records. 
	First, we observe that each edge can appear in at most two records, because it can connect at most two subtrees.
	
	In the next step, we compute a set $D'$, which contains the longest paths which can be donated by $T_v$, for each combination of records. For this we look at the $D$-sets in our records from $\overline{\Recs(v)}, \Recs(u_1), \ldots, \Recs(u_\ell)$. We trace out the longest paths along edges in the $D$-sets of these records, which can be done in time $k^{\bigoh(1)}$ (we start at some edge $e_i$ and find its partner $e_j$ in the same $D$-set, then we look for $e_j$ in the other $D$-sets, and so on; in particular, this is not an ordinary longest-path problem).
	
	Now, we resolve each of the pairs $\{e_i,e_j\}$ in $R$ for any of the currently considered records using the paths in $D'$. Either there is a path in $D'$ connecting $e_i$ and~$e_j$, which means the pair $\{e_i,e_j\}$ can be ignored. 
	Or there are two paths connecting~$e_i$ resp.\ $e_j$ to $e_i'$ resp.\ $e_j' \in E(\partial(v))$.
	Then the pair $\{e_i',e_j'\}$ needs to be added to~$R'$. 
	In case $e_i \in E(\partial(v))$, let $e_i'=e_i$ and similarly for $e_j$.
	In either case the used paths are deleted from $D'$.
	
	Next, we consider each pair $(s,e_i) \in S$ for any of the currently considered records. Let $(s,t) \in P$. 
	If $t \notin \forgottenv_v$ and $e_i \in E(\partial(v))$, then add $(s,e_i)$ to $S'$. In case $e_i \notin E(\partial(v))$,
	we use the donated paths in $D'$ to connect $e_i$ to $e_i' \in E(\partial(v))$ and add $(s,e_i')$ to $S'$. 
	If~$(t,e_j) \in \bar{S}$ for any of the currently considered records, we proceed as if~$(e_i,e_j)\in R$. 
	
	Note that all steps are deterministic, as each edge can only appear in two sets and therefore there can only be one path starting at any edge $e$ that one could use to traverse the graph.
	
	Afterwards, we need to delete all pairs in $D'$ with $e_i$ or $e_j \notin E(\partial(v))$. 
	Finally, the tuple $(S',D',R')$ is inserted as a record in $\Recs(v)$. 
	
	Correctness follows via induction: The records of the leaves are valid. Assuming $\Recs(u_1), \ldots, \Recs(u_\ell)$ are valid, so will be the record set at $v$: It contains all possible ways in which the partial solutions of the subtrees at $u_i$ could be extended. In particular, this includes the projections of all full solutions, and by construction, every such extension will extend the combination of partial solutions of the subtrees to a partial solution of the subtree at $v$, showing validity.
	
	As for the running time: We go through each of the $n$ vertices, where $|\mathcal{R}_v|\leq 2^{\bigoh(k\log k)}$ for $v \in V$. Moreover, each vertex has at most $2(k-1)$ children, which makes for $2^{\bigoh(k^2\log k)}$ combinations when branching, and the number of combinations dominates the time each combination takes to be processed. Hence, the total running time amounts to $2^{\bigoh(k^2\log k)} \cdot n$.
\fi
\qed 
\end{proof}

\iflong
\subsection{List Coloring}
\fi
\ifshort
\smallskip
\noindent \textbf{List Coloring.}\quad
\fi
	The second problem we consider is \textsc{List Coloring}~\cite{FellowsEtAl11,Ganian0S15}. It is known that this problem is \W{1}-hard parameterized by tree-cut width. 
	A \emph{coloring} $\col$ is a mapping from the vertex set of a graph to a set of colors; a coloring is \emph{proper} if for every pair of adjacent vertices $a,b$, it holds that $\col(a) \neq \col(b)$.
	
	\pbDef{\textsc{List Coloring}}
	{A graph $G = (V, E)$ and for each vertex $v \in V$ a list $L(v)$ of permitted colors.}
	{Does $G$ admit a proper coloring \col\ where for each vertex $v$ it holds $\col(v) \in L(v)$?}
	
	\iflong
	\begin{theorem}
	\fi
	\ifshort
	\begin{theorem}[$\star$]	
	\fi
		\textsc{List Coloring} is fixed-parameter tractable when parameterized by the \width\ of the input graph.
	\end{theorem}
\iflong
\begin{proof}
\fi
\ifshort
\begin{proof}[Sketch]
\fi
		We start by defining the syntax of the records we will use in our dynamic program. For $v\in V(G)$, let a record for a vertex $v$ consist of tuples of the form $(u,c)$, where (1) $u \in V(\partial(v))\cap \forgottenv_v$, (2) $c \in L(u) \cup \{\delta\}$, and (3) each vertex of $V(\partial(v))\cap \forgottenv_v$ appears exactly once in a record.
		
%
%
		
		To introduce the semantics of the records, consider the set $\mathcal{W}$ containing all partial solutions (i.e., all proper colorings) at $v$ to the instance $(\forgottenG_v,(L(u))_{u \in \forgottenv_v})$. 
		The $v$-\emph{projection} of a partial solution $\col\in \mathcal{W}$ is a set $R_\col=\{(u,c)~|~u \in V(\partial(v))\cap \forgottenv_v, c \in L(u)\})$ where $(u,c) \in R_\col$ if and only if $\col(u)=c$.
		
		Let $\Recs(v)$ be a set of records for $v$. For two records $R_1,R_2 \in \Recs(v)$ we say $R_1 \preceq R_2$ if and only if for each $u \in V(\partial(v)) \cap \forgottenv_v$ the following holds:
		\begin{itemize}[topsep=0pt]
			\item Either $(u,c) \in R_1 \cap R_2$ with $c \in L(u)$,
			\item Or $(u,c) \in R_1$ with $c \in L(u)$ and $(u,\delta) \in R_2$.
		\end{itemize}
		
		We say that $\Recs(v)$ is \emph{valid} if for each $v$-projection $R_\col$ of a partial solution $\col \in \mathcal{W}$ there is a record $R \in \Recs(v)$ which satisfies $R_\col \preceq R$, and in addition, for every record $R \in \Recs(v)$, there is a partial solution $\col \in \mathcal{W}$ such that its $v$-projection fulfills $R_\col \preceq R$.
		Observe that if $\Recs(r)=\emptyset$, then $(G,(L(v))_{v \in V(G)})$ is a NO-instance, while if $\Recs(r)=\{\emptyset\}$, then $R(r)$ is a YES-instance. 
		
		If a record in $\Recs(v)$ contains a tuple $(u,\delta)$, then this means that there is always a possible coloring for the vertex $u$, e.g., if $|L(u)|>d_G(u)$; the symbol $\delta$ is introduced specifically to bound $|L(v)|$. Therefore, it follows that $|\Recs(v)|\leq 2^{\bigoh(k \log k)}$ for each $v\in V(G)$.		
		To complete the proof, it now suffices to dynamically compute a set of valid records in a leaf-to-root fashion along $T$.
 \iflong
 
	\textbf{If $v$ is a leaf}, we set $\Recs(v) = \{ \{(v,\delta)\}\}$ for the case $|L(v)|>d_G(v)$. Otherwise, we branch over all possible colorings of the vertex $v$, i.e., $\Recs(v) = \{ \{(v,c)\}~|~c \in L(v)\}$.
	Note that the amount of records is always bounded by $k$, as $d_G(v)\leq k$. 
	
	\textbf{If $v$ is an internal node}, we	start with reducing the size of the subtree $T_v$ in the following way: Let $u$ be a child of $v$ with $E(\partial(u))=\{(u,v)\}$.
	\begin{itemize}[topsep=0pt]
		\item If $\Recs(u)= \{\{(u,c)\}\}$ with $c\neq\delta$, then remove $c$ from $L(v)$.
		
		\item Delete $T_u$.
	\end{itemize}
	After this step there are at most $2(k-1)$ children of $v$ left. 
	Let $u_1, \ldots, u_\ell$ with $\ell \leq 2(k-1)$ denote the remaining children of $v$. 
	First, we compute a set $\overline{\Recs(v)}$, in the same way we would compute $\Recs(v)$ if $v$ was a leaf. Our goal is to compute $\Recs(v)$ using the local set $\overline{\Recs(v)}$ and the partial results $\Recs(u_1), \ldots, \Recs(u_\ell)$.
	Note that since $d_G(v)\leq 2k-1$ the number of records in $\overline{\Recs(v)}$ is also bounded by $2k-1$.
	
	In the next step we take one record each out of $\overline{\Recs(v)}, \Recs(u_1), \ldots, \Recs(u_\ell)$ and branch over all possible combination of records. 
	Then we check for each combination if the coloring of the vertices in $\forgottenv_{u_1}, \ldots, \forgottenv_{u_\ell}$ can be combined to a proper coloring of the vertices in $\forgottenv_v$. For this we only need to consider the vertices in $\partial(u_1), \ldots, \partial(u_\ell)$ and check if two neighbors share the same color.
	If this is not possible, then move on to the next combination of records. 
	
	Afterwards, we need to remove all the vertices, which are not in $V(\partial(v))\cap \forgottenv_v$. The remaining vertices and their colors form a record of $T_v$.
	
	
	
	Since $|\Recs(v)|\leq 2^{\bigoh(k \log k)}$ and the size of each record is bounded by $\bigoh(k)$, the running time is bounded by $2^{\bigoh(k^2 \log k)} \cdot n$.
	\fi
	\qed
	\end{proof}

\iflong
\subsection{Boolean CSP}
\fi
\ifshort
\smallskip
\noindent \textbf{Boolean CSP.}\quad
\fi
	Next, we consider the classical constraint satisfaction problem~\cite{SamerSzeider10a}. \ifshort ($\star$)\fi
\iflong
	An instance $I$ of \textsc{Boolean CSP} is a tuple $(X,C)$, where $X$ is a finite set of \emph{variables} and $\mathcal{C}$ is a finite set of \emph{constraints}. Each constraint in $\mathcal{C}$ is a pair $(S,R)$, where the \emph{constraint scope} $S$ is a non-empty sequence of distinct variables of $X$, and the \emph{constraint relation} $R$ is a relation over $\{0,1\}$ (given as a set of tuples) whose arity matches the length of $S$. An \emph{assignment} is a mapping from the set $X$ of variables to $\{0,1\}$.
	An assignment $\sigma$ satisfies a constraint $C = ((x_1, \ldots, x_n),R)$ if $(\sigma(x_1), \ldots, \sigma(x_n)) \in R$, and $\sigma$ satisfies the \textsc{Boolean CSP} instance if it satisfies all its constraints. An instance~$I$ is satisfiable if it is satisfied by some assignment.
	\fi
	
	\pbDef{\textsc{Boolean CSP}}
	{A set of variables $X$ and a set of constraints $\mathcal{C}$.}
	{Is there an assignment $\sigma: X \rightarrow \{0,1\}$ such that all constraints in $\mathcal{C}$ are satisfied?}
	
	We represent this problem via the \emph{incidence graph}, whose vertex set is $X\cup \mathcal{C}$ and which contains an edge between a variable and a constraint if and only if the variable appears in the scope of the constraint.
	
	\iflong
	\begin{theorem}
	\fi
	\ifshort
	\begin{theorem}[$\star$]	
	\fi
		\textsc{Boolean CSP} is fixed-parameter tractable when parameterized by the \width\ of the incidence graph.
	\end{theorem}
	
	\iflong
	\begin{proof}
		\fi
	\ifshort
	\begin{proof}[Sketch]	
	\fi		 		
		For this problem, we do not need to consider all the vertices in the boundary. Instead, for a vertex $v \in V$, let $B(v)=V(\partial(v))\cap \forgottenv_v \cap X$. 
	Hence, we will consider only the vertices in the boundary inside of the current subtree, which correspond to variables in the input instance. Note that $|B(v)|\leq |V(\partial(v))|\leq 2k$.	

				
		We continue with defining the syntax of the records we will use in our dynamic program. For $v\in V(G)$, let a record for a vertex $v$ be a set of functions of the form $\varphi:B(v)\rightarrow\{0,1\}$.		
		Let $\Recs(v)$ be a set of records for $v$. 
		From the syntax, it follows that $|\Recs(v)|\leq 2^{\bigoh(k)}$ for each $v\in V(G)$. 
		To introduce the semantics of the records, consider the set $\mathcal{W}$ containing all partial solutions (i.e., all assignments of the variables such that every constraint is fulfilled) at $v$ for the instance $(\forgottenv_v \cap X,\forgottenv_v \cap \mathcal{C})$.
						
		The function $\varphi$ is a $v$-\emph{projection} of a solution $\sigma\in \mathcal{W}$ if and only if $\sigma \raise-.5ex\hbox{\ensuremath|}_{B(v)} = \varphi$.
		This means, that the functions in a record represent the assignments of variables, which are compatible with $\forgottenv_v$.
		
		We say that $\Recs(v)$ is \emph{valid} if it contains all $v$-projections of partial solutions in~$\mathcal{W}$, and in addition, for every record in $\Recs(v)$, there is a partial solution such that its $v$-projection yields this record.
		Observe that if $\Recs(r)=\emptyset$, then $(X,\mathcal{C})$ is a NO-instance, while if $\Recs(r)=\{\emptyset\}$, then $R(r)$ is a YES-instance. To complete the proof, it now suffices to dynamically compute a set of valid records in a leaf-to-root fashion along $T$.
		\iflong

		\textbf{If $v$ is a leaf and $v \in X$}, we can remove $v$ in case $d_G(v)=1$. Otherwise, we set $B(v)=\{v\}$ and all assignments are valid, i.e., $\Recs(v)=\{\varphi: \{v\} \rightarrow \{0,1\}\}$.
%
		
		\textbf{If $v$ is a leaf and $v \in C$}, then $B(v)=\emptyset$, which means $\Recs(v)=\{\emptyset\}$.
%
%
		
		\textbf{If $v$ is an internal node}, we start with bounding the number of children of $v$. We have to distinguish, if $v$ corresponds to a variable or a constraint. Let $u$ be a child of $v$. 
		\begin{itemize}[topsep=0pt]
			\item For $v \in X$ and $B(u)=\emptyset$, check whether $\Recs(u)$ allows both values for the variable~$v$. If not we fix the value as seen in the previous case. Afterwards delete $u$.
			
			\item For $v \in \mathcal{C}$ and  $B(u)=\{u\}$, use $\Recs(u)$ to check all viable assignments to the root and then remove the unsatisfiable ones from the constraint $v$. Afterwards delete $u$.
			
			\item If after this we obtain an empty constraint or a conflict with the variable assignment occurs, we know that this is a NO-instance.
		\end{itemize}
		After this step there are at most $2(k-1)$ children left. Let $u_1, \ldots, u_\ell$ with $\ell \leq 2(k-1)$ denote the remaining children of $v$. To obtain $\Recs(v)$, we can brute force all viable combinations of $\Recs(u_1), \ldots, \Recs(u_{\ell})$.
		
		Since the number of records and the size of each record is bounded by $k$, this algorithm runs in time $2^{\bigoh(k^2)}\cdot n$. 
%
%
%
%
\fi
		\qed
	\end{proof}

\ifshort
\smallskip
\noindent \textbf{Further Problems.}\quad
As our final two results, we use the algorithmic framework developed above to also establish the fixed-parameter tractability for the remaining two problems which were shown to be \W{1}-hard w.r.t.\ tree-cut width. These are \textsc{Maximum Stable Roommates with Ties and Incomplete Lists (MaxSRTI)}~\cite{BredereckHKN19} and \textsc{Minimum Changeover Cost Arborescence (MinCCA)}~\cite{GozupekOPSS17}. ($\star$)


	\begin{theorem}[$\star$]
		\textsc{MaxSRTI} is fixed-parameter tractable when parameterized by the \width\ of the acceptability graph.
	\end{theorem}
	

		\begin{theorem}[$\star$]
		\textsc{MinCCA} is fixed-parameter tractable when parameterized by the \width\ of the input graph.
	\end{theorem}
\fi

\iflong
\subsection{Maximum Stable Roommates with Ties and Incomplete Lists} 
	Our fourth problem originates from the area of computational social choice~\cite{BredereckHKN19}.
	In this problem we are given a set of \emph{agents} $V$, where each agent $v \in V$ has a \emph{preference} $\mathcal{P}_v = (P_v, \preceq_v)$. 
	The agents $P_v \subseteq V \setminus \{v\}$ are called \emph{acceptable (for $v$)} and $\preceq_v$ is a linear order on $P_v$ with ties. Let $u, w\in P_v$. If $u \prec_v w$ then we say that $v$ \emph{strongly} prefers $u$ to $w$;  on the other hand, if $u \prec_v w$ does not hold then we say that $v$ \emph{weakly prefers} $w$ to $u$.

	We represent this problem via the undirected \emph{acceptability graph} $G$, which contains a vertex for each agent in $V$ and an edge between two agents if and only if both appear in the preference lists of the other. 
	
	A set $M \subseteq E(G)$ is called a \emph{matching} if no two edges in $M$ share an endpoint. If the edge $\{v,w\}$ is contained in $M$, then we say $v$ is \emph{matched} to $w$ and denote this as $M(v)=w$ and vice versa. In case a vertex $v$ is not incident to any edge in $M$, then $v$ is \emph{unmatched} resp.\ $M(v)=\bot$ (where we assume $\bot$ to be less preferable than all acceptable neighbors of $v$). 
	An edge $\{v,w\} \in E(G)\setminus M$ is \emph{blocking} for $M$ (we also say \ $v,w$ form a \emph{blocking pair}) if $w \prec_v M(v)$ and $v \prec_w M(w)$. A matching is \emph{stable} if it does not admit a blocking pair.
	
	\pbDef{\textsc{Maximum Stable Roommates with Ties and Incomplete Lists (MaxSRTI)}}
	{A set of agents $V$, a preference profile $\mathcal{P} = (\mathcal{P}_v)_{v\in V}$, and an integer $\pi$.}
	{Is there a stable matching of $(V,P)$ of cardinality at least $\pi$?}
	
	
	\begin{theorem}
		\textsc{MaxSRTI} is fixed-parameter tractable when parameterized by the \width\ of the acceptability graph.
	\end{theorem}

\newcommand{\mat}{\texttt{matched}}
\newcommand{\uns}{\texttt{unsafe}}
\newcommand{\saf}{\texttt{safe}}
\newcommand{\forgottenH}{\mathcal{H}}
\newcommand{\sig}{\texttt{sig}}
	\begin{proof}	
We once again start by defining the syntax of the records. For $v\in V(G)$, let a \emph{signature} at $v$ be a mapping $E(\partial(v))\rightarrow \{\mat,\uns,\saf\}$. Clearly, the number of signatures at $v$ is upper-bounded by $3^k$, where $k=\widthshort(G)$.

To make it easier to describe the semantics of the records, let us first define the graph $\forgottenH_v$ as the non-disjoint union of $\forgottenG_v$ and $\partial(v)$; we recall that $\partial(v)$ contains both vertices in $\forgottenG_v$ and vertices adjacent to these, and that $E(\partial(v))$ forms an edge-cut separating $\forgottenG_v$ from the rest of $G$.


We are now ready to define the semantics of the records. 
A matching $M$ in $\forgottenH_v$ is called a \emph{partial solution} if there is no blocking edge for $M$ in $E(\forgottenH_v)$; in other words, we explicitly forbid the edges in the boundary from forming blocking pairs in partial solutions. Each partial solution $M$ \emph{corresponds} to a signature $\sig$ at $v$ defined as follows:
\begin{itemize}
\item for each $e\in M\cap E(\partial(v))$, $\sig(e)=\mat$, 
\item for each $e=ab\in E(\partial(v))\setminus M$ such that $a\in \forgottenv_v$ and there exists $ac\in M$ such that $c \prec_a b$ , $\sig(e)=\saf$, and
\item $\sig(e)=\uns$ otherwise.
\end{itemize}

Intuitively, the signature of $M$ captures the following information about $M$: which edges in the boundary are matched, and for those which are not matched it stores whether they are ``safe'' (meaning that the endpoint in $\forgottenv_v$ will never form a blocking pair with that edge), or  ``unsafe''  (meaning that the endpoint in $\forgottenv_v$ could later form a blocking pair with that edge, depending on the preferences and matching of the endpoint outside of $\forgottenv_v$).

We define $\Rec(v)$ to be a mapping from the set of all signatures at $v$ to $\Nat\cup \{-\infty\}$, where (1) if there is no partial solution corresponding to a signature $\tau$, then $\Rec(v)(\tau)\mapsto -\infty$, and otherwise (2) $\Rec(v)$ maps $\tau$ to the size of the largest partial solution in $\forgottenH_v$ whose signature is $\tau$. To avoid any confusion, we remark that when applying addition to the images of $\Rec(v)$, we let $-\infty + x = -\infty$ for each $x\in \Nat \cup \{-\infty\}$.

If we can compute $\Rec(r)$ for the root $r$ of a spanning tree $T$ witnessing that $\widthshort(G,T)\leq k$, then by definition each partial solution is also a stable matching in the instance. Hence, it suffices to check whether $\Rec(r)(\emptyset)\geq \pi$; if this is the case then we output ``Yes'', and otherwise we can safely output ``No''. At this point, it suffices to compute $\Rec(v)$ for each $v\in V(G)$ in leaf-to-root fashion along $T$. 

\smallskip
	\textbf{If $v$ is a leaf}, we first add the mapping $(E(\partial(v))\mapsto \uns)\mapsto 0$ to $\Rec(v)$, which corresponds to the empty partial solution. Then, for each $vw\in E(\partial(v))$ we construct a signature $\tau_w$ which assigns $vw$ to $\mat$, and for each neighbor $u$ of $v$ other than $w$ either assigns $vu$ to $\saf$ (if $v$ weakly prefers $w$ to $u$) or assigns $vu$ to $\uns$ (if $v$ strongly prefers $u$ to $w$). For each $\tau_w$ constructed in this way, we set $\Rec(v)(\tau_w)=1$.
	
	\smallskip
	\textbf{If $v$ is an internal node}, we begin by branching over all edges incident to $v$, and for each such edge $vw$ we proceed by restricting our attention to all partial solutions which contain $vw$. We also have a separate branch to deal with all partial solutions where $v$ remains unmatched; we will begin by dealing with this (slightly simpler) case. 

	\smallskip	
	\textbf{Subcase: $v$ remains unmatched}. For each child $u$ of $v$ such that $E(\partial(u))=\{(u,v)\}$, we observe that only partial solutions at $u$ with the signature $uv\mapsto \saf$ can be extended to a partial solution at $v$; indeed, $uv\mapsto \mat$ would violate our assumption that $v$ remains unmatched, while $uv\mapsto \uns$ would, by definition, lead to a blocking pair. For brevity, let us set $\texttt{simple-size}$ to be the sum of all $\Rec(u)(uv\mapsto \uns)$ over all vertices $u$ with a single-edge boundary.
	
	As in the previous algorithms, we observe that at this point only at most $2k$ children of $v$ remain to be processed, say $x_0,\dots,x_\ell$. We proceed by simultaneously branching over all of the at most $3^k$ signatures for each of these children, resulting in a total branching factor of $3^{k^2}$; each branch can be represented as a tuple $(\sig_{x_0},\dots,\sig_{x_\ell})$. We now discard all tuples that are not well-formed, where a tuple is well-formed if the following conditions hold: 
	\begin{itemize}[topsep=0pt]
\item it contains no signature that maps an edge incident to $v$ to either $\uns$ or $\mat$ (as before, these edges may only be mapped to $\saf$);
\item for each edge $ab$ such that $a\in \forgottenv_{x_i}$ and $b\in \forgottenv_{x_j}$, $i,j\in [\ell]$, the signatures of $x_i$ and $x_j$ must either (a) both map that edge to $\mat$, or (b) both map that edge to $\saf$, or (c) map that edge to $\saf$ once and $\uns$ once (signatures must be consistent).
\end{itemize}

For all remaining tuples, we set \texttt{branching-size} to $\sum_{i\in [\ell]}\Rec(x_i)(\sig_{x_i})$. We also identify a unique signature $\sig^*$ corresponding to the current branch as follows: each edge in $\partial(v)$ incident to $v$ is mapped to $\uns$, and each edge $e$ in $\partial(v)$ not incident to $v$ must have an endpoint in $\forgottenv_{x_i}$ for some $x_i$ and is mapped to $\sig_{x_i}(e)$. At this point, we update $\Rec(v)(\sig^*)$ as follows: if the value of $\Rec(v)(\sig^*)$ computed so far is greater than $\texttt{simple-size}+\texttt{branching-size}$ then we do nothing, and otherwise we set that value to $\texttt{simple-size}+\texttt{branching-size}$. We now proceed to the next branch, i.e., choice of neighbor of $v$.

	\smallskip
	\textbf{Subcase: $v$ is matched to $w$}. We will in principle follow the same steps as in the previous subcase, but with a few extra complications. Let us begin by distinguishing whether (1) $w$ itself is a child of $v$ such that $E(\partial(u))=\{(u,v)\}$, (2) $w$ is in $\forgottenv_{x_i}$ for some child $x_i$ of $v$ not satisfying this property (including the case where $w=x_i$), or (3) $w\not \in \forgottenv_v$. In the first case, we set the child $w$ aside and initiate $\texttt{simple-size}=\Rec(w)(wv\mapsto \mat)$. In the second case, we will later (in the appropriate branching step) discard all signatures of $x_i$ which do not map $wv$ to $\mat$. In the third case, we will take this into account when constructing $\sig^*$.
	
	Next, for each child $u$ of $v$ such that $E(\partial(u))=\{(u,v)\}$ (other than $w$, in case (1)), we distinguish whether $v$ weakly prefers $w$ to $u$, or not. For each $u$ where this holds, we observe that any partial solution at $u$ that does not use $v$ can be safely extended to a partial solution at $v$---hence, we increase $\texttt{simple-size}$ by $\max(\Rec(u)(vu\mapsto \uns), \Rec(u)(vu\mapsto \saf))$. On the other hand, for each $u$ where $v$ strongly prefers $u$ to $w$ we observe that a partial solution at $u$ can only be extended to one at $v$ if it matches $u$ in a way which prevents the creation of a blocking pair with $v$. Hence, in this case, we increase  $\texttt{simple-size}$ by $\Rec(u)(vu\mapsto \saf)$.
	
	In the second step, we once again proceed by simultaneously branching over all of the at most $3^k$ signatures for the remaining children $x_0,\dots,x_\ell$ of $v$. As before, this results in a total branching factor of $3^{k^2}$, and each branch can be represented as a tuple $(\sig_{x_0},\dots,\sig_{x_\ell})$. We now discard all tuples that aren't well-formed, where a tuple is well-formed if the following conditions hold:
	\begin{itemize}[topsep=0pt]
\item for each edge $ab$ such that $a\in \forgottenv_{x_i}$ and $b\in \forgottenv_{x_j}$, $i,j\in [\ell]$, the signatures of $x_i$ and $x_j$ must either (a) both map that edge to $\mat$, or (b) both map that edge to $\saf$, or (c) map that edge to $\saf$ once and $\uns$ once (signatures must be consistent);
\item in case (2), the edge $vw$ is mapped to $\mat$ in the appropriate signature;
\item the tuple contains no signature that maps any edge incident to $v$ (other than $vw$) to $\mat$;
\item for no edge $vz$ where $z\in \forgottenv_{x_i}$ for some $i\in [\ell]$ such that $v$ strongly prefers $z$ to $w$, the signature of $x_i$ maps $vz$ to $\uns$ (as this would create a blocking pair).
\end{itemize}

For all remaining tuples, we set \texttt{branching-size} to $\sum_{i\in [\ell]}\Rec(x_i)(\sig_{x_i})$ in cases (1) and (2); in case (3), we set it to $\sum_{i\in [\ell]}\Rec(x_i)(\sig_{x_i})+1$. We also identify a unique signature $\sig^*$ corresponding to the current branch as follows: each edge $vc\in E(\partial(v))$ is mapped to $\uns$ if $v$ strongly prefers $c$ to $w$, and $\saf$ otherwise (with the exception of $c=w$ in case (3), where $vw$ must be mapped to $\mat$). Furthermore, each edge $e$ in $\partial(v)$ not incident to $v$ must have an endpoint in $\forgottenv_{x_i}$ for some $x_i$ and is mapped to $\sig_{x_i}(e)$. At this point, we update $\Rec(v)(\sig^*)$ as follows: if the value of $\Rec(v)(\sig^*)$ computed so far is greater than $\texttt{simple-size}+\texttt{branching-size}$ then we do nothing, and otherwise we set that value to $\texttt{simple-size}+\texttt{branching-size}$. We then proceed to the next branch, i.e., choice of neighbor of $v$.

The correctness of the algorithm can be shown by induction; it is not difficult to verify that the computation of the records is correct at the leaves, and for non-leaves one uses the assumption that the records of the children are correct. The crucial point is that every partial solution at a child that corresponds to a certain signature can be extended to a partial solution at the parent if the verified conditions hold, which justifies the correctness of adding up the appropriate values for the children. The running time is upper-bounded by $3^{k^2}\cdot n^2$.
		\qed
	\end{proof}

\subsection{Minimum Changeover Cost Arborescence}
	The final problem we consider can be found in~\cite{GozupekOPSS17}. An \emph{arborescence} is a directed tree with root $r$, which contains a directed path from each vertex to $r$.
	
	Given an arborescence $T$ with root $r$ and an edge $e \in E(T)$ we denote with $\suc(e)$ the edge incident to $e$ on the path from $v$ to the root $r$.
	For an edge $e$ incident to the root we define $\suc(e) = e$. 
	
	A function $\cost: X^2 \rightarrow \Nat$ is called a \emph{changeover cost function} if it satisfies the following:
	\begin{enumerate}
		\item $\cost(x_1,x_2) = \cost(x_2,x_1)$ for each $x_1,x_2 \in X$, and
		\item $\cost(x,x)=0$ for each $x \in X$. 
	\end{enumerate}
	
	The \emph{total changeover costs} of an arborescence $T$ are now defined as
	\begin{align*}
		\sum_{e \in E(T)} \cost(e,\suc(e)).
	\end{align*}
	
	\pbDef{\textsc{Minimum Changeover Cost Arborescence (MinCCA)}}
	{A directed graph $G=(V,E)$, a root $r\in V(G)$, an edge coloring $\col: E(G) \rightarrow X$, and a changeover cost function $\cost: X^2 \rightarrow \Nat$.}
	{What is an arborescence of $G$ minimizing the total changeover costs?}
	
	The \width\ of a directed graph $G$ is the \width\ of $G$ where we omit the arc directions. 
	
	\begin{theorem}
		\textsc{MinCCA} is fixed-parameter tractable when parameterized by the \width\ of the input graph.
	\end{theorem}

	\begin{proof}
		We start by defining the syntax of the records we will use in our dynamic program. For $v\in V(G)$, let a record for a vertex $v$ be a tuple of the form $(\outgoing, \donate)$, where: 
		\begin{itemize}[topsep=0pt]
			\item $\outgoing = \{(v_0, e_0), \ldots, (v_i, e_i)\}$ where for each $j \in [i]$, $v_j \in V(\partial(v))\cap \forgottenv_v$, $e_j \in E(\partial(v))$, and
			
			\item $\donate = \{(v_0,c_0,e_0), \ldots, (v_i,c_i,e_i)\}$ where for each $j \in [i]$, $v_j \in V(\partial(v))\cap \forgottenv_v$, $c_j \in X$, and $e_j \in E(\partial(v))$.
		\end{itemize}
		Moreover, let $f:\Recs(v) \rightarrow \Nat$ be a function. 
		
		Let $\Recs(v)$ be the set of records for $v$.  		
		
		
		To introduce the semantics of the records, we need the following notion: A \emph{partial solution} at $v$ is a forest of $\forgottenG_v \cup \partial(v)$, where for each vertex $u \in \forgottenv_v$ there is a directed path from $u$ to exactly one vertex in $(V(\partial(v)) \setminus \forgottenv_v) \cup \{r\}$. 
		Consider the set $\mathcal{W}$ containing all partial solutions at $v$. The $v$-\emph{projection} of a partial solution $S \in \mathcal{W}$ is a tuple $(\outgoing_S,\donate_S)$ where:
		\begin{itemize}[topsep=0pt]
			\item $(u,e) \in \outgoing_S$ if and only if there is a $u$-$u'$ path in $S$ with $u' \in V(\partial(v))\setminus \forgottenv_v$ and
			$e$ is the first edge on this path which is contained in $E(\partial(v))$, and
			
			\item $(u_1,c,e) \in \donate_S$ if and only if there exists a path $u_0,u_1,u_2,\ldots, u_{i-1}, u_i$ in~$S$ with $u_0, u_i \in V(\partial(v))\setminus\forgottenv_v$ and $e=(u_{i-1}, u_i)$ and $c = \col(u_1,u_2)$.
		\end{itemize}		
		For a record $R \in \Recs(v)$ the value $f(R)$ denotes the minimum cost of this record, i.e.,
		\begin{align*}
			f(R) = \min_{\substack{S \in \mathcal{W},\\ R= (\outgoing_S,\donate_S)}} \sum_{e \in E(S)} \cost(e,\suc(e)).
		\end{align*}
		
		We say that $\Recs(v)$ is \emph{valid} if it contains all $v$-projections of solutions in $\mathcal{W}$, and in addition, for every record in $\Recs(v)$, there is a partial solution such that its $v$-projection yields this record. 		
		Observe that if $\Recs(r)=\emptyset$, then $(G,r,\col,\cost)$ is a NO-instance, while if $\Recs(r)=\{(\emptyset,\emptyset,\emptyset)\}$, then $R(r)$ is a YES-instance. 
		
		From the syntax and semantics, it follows that $|\Recs(v)|\leq 2^{\bigoh(k \log k)}$ for each $v\in V(G)$.
		
		To complete the proof, it now suffices to dynamically compute a set of valid records in a leaf-to-root fashion along $T$. 
		
		\textbf{If $v$ is a leaf}, we create the following two records for each edge $e\in E(\partial(v))$ outgoing from $v$:
		\begin{itemize}[topsep=0pt]
			\item $\{\{(v,e)\}, \emptyset\}$,
			\item $\{\{(v,e)\}, \{(v,e,\col(e))\}\}$.
		\end{itemize}
		It follows that $f(R)=0$ for each $R \in \Recs(v)$. 
		
		\newcommand{\del}{\textit{del}}
		\textbf{If $v$ is an internal node}, we start with bounding the number of children of $v$ in order to bound the number of records, which need to be computed.
		Let $V_{\del}$ denote the set of children of $v$ which do not increase the \width\ of $v$, i.e., for each $u \in V_{\del}$ it holds $E(\partial(v))=\{(u,v)\}$.
		We define the minimum changeover cost of $V_{\del}$ as $\emph{cost}_{\del} = \sum_{u \in V_{\del}}\min_{R \in \Recs(v)} f(R)$.
		Then, we can delete $T_u$ for each $u \in V_\del$. 
		
		
		After this step there are at most $2(k-1)$ children left.		
		Let $u_1, \ldots, u_\ell$ with $\ell \leq 2(k-1)$ denote the remaining children of $v$.
		First, we compute a local set $\overline{\Recs(v)}$, in the same way we would compute the set of records in the leaf case. Note that the number of edges incident to $v$ is bounded by $2k-1$. Hence, $|\overline{\Recs(v)}|\leq 4k-2$.
		Our goal is to compute $\Recs(v)$ using the local set $\overline{\Recs(v)}$ and the partial results $\Recs(u_1), \ldots, \Recs(u_\ell)$. 
		
		In the following we take one record each out of $\overline{\Recs(v)}, \Recs(u_1), \ldots, \Recs(u_\ell)$ and repeat the following process for each combination of records.
		We proceed similarly as in the proof of \EDP\ (Theorem~\ref{thm:EDP}). 
		First, we combine the donated paths by computing a set $\donate'$, which contains the longest paths which can be donated by $T_v$. 
		For this we look at the $\donate$-sets in our records from $\overline{\Recs(v)}, \Recs(u_1), \ldots, \Recs(u_\ell)$. We trace out the longest paths along edges for each vertex $u$ in a tuple $(u,c,e)$ in the $\donate$-sets of these records, which can be done in time $k^{\bigoh(1)}$.
		
		Next, we consider each pair $(u,e) \in \outgoing$ for any of the currently considered records.
		If $e \in E(\partial(v))$, then add $(u,e)$ to $\outgoing'$. In case $e \notin E(\partial(v))$, we use the donated paths in $\donate'$ to connect $e$ to $e' \in E(\partial(v))$, where the sink of $e'$ is in $V(\partial(v))\setminus\forgottenv_v$, and add $(u,e')$ to $\outgoing'$. 
		
		Afterwards, we need to delete all pairs in $\donate'$ with $u \notin V(\partial(v)\cap \forgottenv_v)$ or $e \notin E(\partial(v))$.
		Finally, the tuple $(\outgoing',\donate')$ is inserted as a record in~$\Recs(v)$.		
		Note that all steps are deterministic, as for each vertex there is exactly one outgoing edge. 		
		
		Let $R_1, \ldots, R_\ell$ be the records used to compute $R \in \Recs(v)$. The integer $\emph{cost}_{\emph{conn}}$ denotes the sum of the changeover cost for connecting an outgoing tuple with a longest donate path.
		Now, we can determine the minimum cost of $R \in \Recs(v)$ by computing $f(R)= \min\{f(R),\emph{cost}_{\emph{conn}} + \emph{cost}_{\del} + \sum_{j=1}^{\ell} f(R_j)\}$, where we initiate $f(R)=\infty$.
		
		Since the number of records and the size of each record is bounded by $k$, the running time of this algorithm is $2^{\bigoh(k^2 \log k)}\cdot n$.
		\qed
	\end{proof}
\fi

\section{Conclusion}
The parameter developed in this paper, \width, is aimed at mitigating the algorithmic shortcomings of tree-cut width and filling the role of an ``easy-to-use'' edge-based alternative to treewidth. We show that \width\ essentially has all the desired properties one would wish for as far as algorithmic applications are concerned: it is easy to compute, uses a natural structure as its decomposition, and yields fixed-parameter tractability for all problems that one would hope an edge-based alternative to treewidth could solve.

Last but not least, we note that a preprint exploring a different parameter that is aimed at providing an edge-based alternative to treewidth appeared shortly after the results presented in our paper were obtained~\cite{edgetreewidth}. While it is already clear that the two parameters are not equivalent, it would be interesting to explore the relationship between them in future work.

\bibliographystyle{splncs04}
\bibliography{literature.bib}
\end{document}